\documentclass[reqno,11pt]{amsart}
\usepackage{amsmath, latexsym, amsfonts, amssymb, amsthm, amscd}
\usepackage{ graphics,epsf,psfrag}

\advance\hoffset -.75cm

\numberwithin{equation}{section}

\newtheorem{theo}{Theorem}[section]

\newtheorem{pro}[theo]{Proposition}

\makeatletter
\def\captionfont@{\footnotesize}
\def\captionheadfont@{\scshape}

\long\def\@makecaption#1#2{%
  \vspace{2mm}
  \setbox\@tempboxa\vbox{\color@setgroup
    \advance\hsize-6pc\noindent
    \captionfont@\captionheadfont@#1\@xp\@ifnotempty\@xp
        {\@cdr#2\@nil}{.\captionfont@\upshape\enspace#2}%
    \unskip\kern-6pc\par
    \global\setbox\@ne\lastbox\color@endgroup}%
  \ifhbox\@ne 
    \setbox\@ne\hbox{\unhbox\@ne\unskip\unskip\unpenalty\unkern}%
  \fi
  \ifdim\wd\@tempboxa=\z@ 
    \setbox\@ne\hbox to\columnwidth{\hss\kern-6pc\box\@ne\hss}%
  \else 
    \setbox\@ne\vbox{\unvbox\@tempboxa\parskip\z@skip
        \noindent\unhbox\@ne\advance\hsize-6pc\par}%
\fi
  \ifnum\@tempcnta<64 
    \addvspace\abovecaptionskip
    \moveright 3pc\box\@ne
  \else 
    \moveright 3pc\box\@ne
    \nobreak
    \vskip\belowcaptionskip 
  \fi
\relax
}
\makeatother
\def\writefig#1 #2 #3 {\rlap{\kern #1 truecm
\raise #2 truecm \hbox{#3}}}

\DeclareMathSymbol{\leqslant}{\mathalpha}{AMSa}{"36} 
\DeclareMathSymbol{\geqslant}{\mathalpha}{AMSa}{"3E} 
\DeclareMathSymbol{\eset}{\mathalpha}{AMSb}{"3F}     
\renewcommand{\leq}{\;\leqslant\;}                   
\renewcommand{\geq}{\;\geqslant\;}                   

\newcommand{\bra}{\langle}
\newcommand{\ket}{\rangle}

\newcommand{\cC}{\ensuremath{\mathcal C}}

\newcommand{\cF}{\ensuremath{\mathcal F}}
\newcommand{\cG}{\ensuremath{\mathcal G}}

\newcommand{\cJ}{\ensuremath{\mathcal J}}

\newcommand{\bbP}{{\ensuremath{\mathbb P}} }

\newcommand{\bbZ}{{\ensuremath{\mathbb Z}} }

\newcommand{\ga}{\alpha}
\newcommand{\gb}{\beta}
\newcommand{\gd}{\delta}
\newcommand{\gp}{\varphi}
\newcommand{\gr}{\rho}

\newcommand{\gD}{\Delta}

\newcommand{\gl}{\lambda}

\newcommand{\gs}{\sigma}

\newcommand{\hq}{\hat q}
\newcommand{\grm}{\gr^-}
\newcommand{\grp}{\gr^+}

\newcommand{\mmu}{E}

\title{A diffusive system driven by a battery or by a smoothly varying field}

\author{T. Bodineau$^{(1)}$, B. Derrida$^{(2)}$, 
J. L. Lebowitz$^{(3)}$}

\address{
(1) D\'epartement de math\'ematiques et applications, Ecole Normale
Sup{\'e}rieure, CNRS-UMR 8553, 75230 Paris cedex 05, France}
\address{(2) Laboratoire de Physique Statistique, Ecole Normale
Sup{\'e}rieure, 24 rue Lhomond, 75231 Paris Cedex 05, France}
\address{(3) Departments of Mathematics and Physics,  Rutgers University, 110 Frelinghuysen Road, Piscataway, NJ 08854}

\date{\today}

       
\begin{document}

\maketitle

\begin{abstract}
We consider the steady state of a one dimensional diffusive system,   such as the symmetric simple exclusion process (SSEP) on a ring, driven by a battery  at the origin or by a smoothly varying field along the ring. The battery appears as the limiting case of a smoothly varying field, when the field becomes a delta function at the origin.
We find that in the scaling limit, 
the long range pair correlation functions of the system driven by a battery turn out to be very different from  the ones known in the steady state of the SSEP maintained out of equilibrium by contact with two reservoirs, even when the steady state density profiles are identical in both models.
\end{abstract}

\section{Introduction}

There are notoriously few fully analyzable models of non trivial (interacting) current carrying systems \cite{L, Schutz} in non-equilibrium steady states (NESS). The few exceptions are almost all one dimensional lattice gases evolving according to stochatic jump processes and interacting via exclusions, in contact with particle reservoirs at different densities.
Among the simplest of such models is one of particles on a lattice of $N$ sites ($i =1, \dots, N$) in which the bulk dynamics are given by the symmetric simple exclusion process (SSEP) and sites $1$ and $N$ are in contact with particle reservoirs. In the case of reservoirs at equal densities the stationary state satisfies detailed balance. For reservoirs at unequal densities, say $\gr_a$ on the left and $\gr_b$ on the right ($\gr_a > \gr_b$), the matrix method gives a full microscopic description of the NESS including the explicit form of the correlation functions 
$\gr_k(i_1,\dots, i_k; \gr_a,\gr_b,N)$ \cite{DLS} and allows one to calculate the large deviation function (LDF) in the hydrodynamical scaling limit,
$N \to \infty$, $i/N \to x \in [0,1]$ \cite{DLS0,BDGJL2, BDGJL5}. One can also obtain, in this scaling limit (with diffusively scaled time), the time evolution of the system in a non stationary state via a diffusion equation for typical density profiles as well as the LDF for time evolving densities and currents \cite{BDGJL2, BDGJL2006, D2, BD2007}.

\medskip

In attempts to further extend our understanding of such non-equilibrium systems, we investigate 
here the stationary state of a system of $M$ particles on a ring of $N$ sites with exclusion.
The jump rates to the right and left (counter clockwise, clockwise) for a particle at site $i$ are $p_i$ and $q_i$, $i = 1 \dots N$.
We shall be interested, as usual, in the case $N \gg 1$, $M/N = \bar \gr$.
We consider two types of situations: 
\begin{enumerate}
\item The "battery" case: the system evolves according to the SSEP except that there is a "battery" at bond $(N,1)$, i.e. $p_i = 1$ for $i \not = N$, $q_i = 1$ for $i \not = 1$ and $p_N = p, q_1 = q$ with $p$ and $q$ independent of $N$.
\item The smooth asymmetric case (WASEP) $p_i = 1 + \frac{1}{N} E (i/N)$ and $q_i = 1 - \frac{1}{N} E (i/N)$ with $E$ a smooth field of period 1.
\end{enumerate}
(Combination of both cases can also be treated).

In the first case, the microscopic "battery" at $(N,1)$ induces a NESS with a particle current of magnitude proportional to $1/N$ and a concomitant linear density profile.
This looks very similar to what happens in the open system and would suggest that the measures for the reservoirs and 
battery driven systems would be similar when $N \gg 1$, in the same spirit as the difference between 
``canonical'' and ``grand-canonical'' nature of the two systems disappears in equilibrium, 
when $p=q$ and $M/N \to \gr_a=\gr_b$.
This is indeed the case to the leading order in $N$.   
Local equilibrium holds in both models and implies that both steady state statistics are given locally by a product Bernoulli measure. 
Here, we will focus on the corrections to this local Bernoulli measure and on the long range correlations. 
We will see that the two NESS measures differ much more than the corresponding equilibrium ones. 
The complicated interplay between the non-equilibrium long-range correlations, the canonical constraint and the driving mechanism leads to 
two-point correlation functions with different structure for the reservoir and battery driven systems.
In particular, the long range correlations are much more singular for the battery model than those in the open system.
A variant of this model with a macroscopic battery was introduced in \cite{HF} and a related dynamics with colliding hard spheres was also considered in \cite{Burdzy}.

In the second case, the macroscopic stationary density profile $\bar \gr(x)$ then satisfies the general equation for driven diffusive systems
\begin{eqnarray}
\label{eq: debut}
\partial_x  \bar \gr(x)  - 2 E(x) \bar \gr(x) \big( 1 - \bar \gr(x) \big) = - \cJ \, ,
\end{eqnarray}
where $\cJ$ is the stationary current which is independent of $x$.
Equation \eqref{eq: debut} corresponds to the stationary solution of the diffusion equation with a drift
\begin{eqnarray}
\label{eq: debut 1}
\partial_t \gr  = \partial^2_x   \gr   - \partial_x  \left( 2 E(x) \gr(1 - \gr) \right) \, .
\end{eqnarray}
The first case can be thought of as a very singular limit of the second one when $E(x)$ becomes a $\delta$-function localized at the origin.

\medskip

We will also consider the  Zero Range Process (ZRP) driven by a "battery".
This microscopic dynamics is simpler than the SSEP driven by a battery, as the NESS of the ZRP on the ring is a product measure constrained to having a fixed number of particles.
There are no long range correlations in the non-equilibrium measure of the ZRP with reservoirs \cite{MF}
and the long range correlations of the ZRP driven by a battery are similar to those found for an equilibrium system with 
the constraint that the total particle number is fixed.

\medskip

The outline of the rest of the paper is as follows.
In section \ref{sec: The models}, we define the microscopic models.
By using a macroscopic approach, the two-point correlation functions are computed for the battery model in section 
\ref{sec: Macroscopic approach} and for a slowly varying field in section \ref{sec: The macroscopic approach for the slowly varying field}.
These results are then compared, in section \ref{sec: Comparison with the open system and with numerical simulations}, 
to the two-point correlation functions of an open system in contact with reservoirs.
Finally the invariant measure of the ZRP driven by a battery is computed in section \ref{sec: exactly solvable models}.
Some technical details concerning the solution of the macroscopic equations are given in the appendices.

\section{The models}
\label{sec: The models}

\subsection{The battery model}
\label{The battery model}

We consider the SSEP on the ring $\{1,N\}$ containing $M$ particles with jump rates $1$ to the 
left and to the right except at the bond $(N,1)$ where the jump rate from $N$ to $1$ is $p$
and from $1$ to $N$ is $q$. These modified asymmetric  rates act as a battery which forces a current of particles through the system. We assume that $p>0$ and $q>0$. 

In the absence of the battery ($p = q = 1$),  the stationary state is one in which all the  $\Big( {N \atop M} \Big)$ configurations have equal weight.
For $N \to \infty$, $M = \bar \gr N$, this corresponds locally to a product Bernoulli measure with density $\bar \gr$. 
Denoting by $\nu_i^{\bar \gr}$ the Bernoulli measure at site $i$ with density
$\bar \gr$, the product measure  $\bigotimes_{i \in \bbZ} \nu_i^{\bar  \gr}$ is invariant wrt the SSEP dynamics for any constant density $\bar \gr$.

For the SSEP on $\bbZ$ with a battery at the bond $(0,1)$, the invariant measure remains a product but the density is discontinuous across 
the battery. The invariant measures on $\bbZ$ are
\begin{eqnarray}
\label{eq: invariant Z}
\nu^{\gr^-,\gr^+} = \bigotimes_{i \leq 0} \nu_i^{\gr^-} \; \bigotimes_{i \geq 1} \nu_i^{\gr^+} \, ,
\end{eqnarray}
where the densities 
$\gr^-,\gr^+$ (at sites  $0$ and $1$) satisfy the equation
\begin{eqnarray}
\label{eq: constraint}
p \; \gr^- (1-\gr^+) = q \; \gr^+ (1- \gr^-) \, .
\end{eqnarray}
These measures satisfy detailed balance and there is no current flowing in the system.

\medskip

On the ring when $p \neq q$, the steady state measure is not known. Taking averages with respect to the time evolving measure
$\mu_{N,\tau}$, we obtain
\begin{eqnarray}
\label{eq: dynamic eq}
i \neq 1,N, \quad &&
\partial_\tau \mu_{N,\tau} \big( \eta_i \big) = \mu_{N,\tau} \big( \eta_{i+1} \big)+ \mu_{N,\tau} \big( \eta_{i-1} \big) - 2 \mu_{N,\tau}  \big( \eta_i \big) \\
&& 
\partial_\tau \mu_{N,\tau} \big( \eta_1 \big) = \mu_{N,\tau} \big( \eta_2 \big) - \mu_{N,\tau} \big( \eta_1 \big)
+ p \mu_{N,\tau} \big( \eta_N (1-\eta_1) \big) - q \mu_{N,\tau} \big( \eta_1 (1-\eta_N) \big)  \nonumber \\
&& 
\partial_\tau \mu_{N,\tau} \big( \eta_N \big) = \mu_{N,\tau} \big( \eta_{N-1} \big) - \mu_{N,\tau} \big( \eta_N \big)
- p \mu_{N,\tau} \big( \eta_N (1-\eta_1) \big) + q \mu_{N,\tau} \big( \eta_1 (1-\eta_N) \big) \nonumber 
\end{eqnarray}
It follows immediately that in the stationary state,  the average density profile is linear 
\begin{eqnarray}
\label{eq: profile steady}
\bra \eta_i \ket = \bra \eta_1 \ket + \frac{i-1}{N - 1} 
\big( \bra \eta_N \ket -  \bra \eta_1 \ket \big) \, ,
\end{eqnarray}
where $\bra \cdot \ket$ stands for the stationary  measure and the mean density $\bar \gr = \frac{M}{N}$ satisfies
\begin{eqnarray}
\label{eq: mean density}
\bar \gr = \frac{1}{2}  \big(\bra \eta_N \ket +  \bra \eta_1 \ket \big) \, .
\end{eqnarray}
There is also a stationary average current $\hq$ of order $1/N$
\begin{eqnarray}
\label{eq: mean flux}  
\hq = \frac{1}{N}  \big( \bra \eta_1 \ket - \bra \eta_N \ket  \big) \, ,
\end{eqnarray}
and one has the identity
\begin{eqnarray}
\label{eq: identity}
p \bra  \eta_N (1-\eta_1) \ket - q \bra \eta_1 (1-\eta_N) \ket
= \hq \, .
\end{eqnarray}
For finite $N$, the relations \eqref{eq: profile steady} and \eqref{eq: identity} do not allow to determine the average profile $\bra \eta_i \ket$ as 
 \eqref{eq: identity} involves the correlation between sites $N$ and $1$. 
If one assumes however that local equilibrium holds for large  $N$ then  the system will behave at the battery as if it was described by the measure 
\eqref{eq: invariant Z}. Thus one expects that for large  $N$
\begin{eqnarray}
\label{eq: decouplage}
\bra \eta_N  \eta_1 \ket \approx \bra \eta_1 \ket \, \bra \eta_N \ket 
= \gr^+ \gr^-\, ,
\end{eqnarray}
with $\gr^\pm$ satisfying 
\begin{eqnarray}
\label{eq: gr+ gr-}
p \; \gr^- (1-\gr^+) = q \; \gr^+ (1- \gr^-), 
\qquad  
\frac{\gr^+ + \gr^-}{2} = \bar \gr \, .
\end{eqnarray}
where we used \eqref{eq: mean density} and \eqref{eq: identity}.
This implies that in terms of the macroscopic variable $x = i/N$, that the steady state density $\bar \gr(x)$ is linear with a discontinuity at the battery
\begin{eqnarray}
\label{eq: bar rho}
\bar \rho(x) = \gr^+ + (\gr^- - \gr^+) x, \quad {\rm with}  \qquad 
\bar \rho(0) = \gr^+,  \quad 
\bar \rho(1) =  \gr^- \, .
\end{eqnarray}

The local equilibrium \eqref{eq: decouplage} can be justified  \cite{BL} and 
one can prove that after space/time rescaling, the microscopic equations lead to a macroscopic description of the local density by the heat equation. 
The battery yields non-linear boundary conditions
\begin{eqnarray}
\label{eq: heat equation}
\forall t>0, x \in ]0,1[, \qquad
\begin{cases}
\partial_t \gr(t,x) = \partial_x^2 \gr(t,x)\\
p \; \gr(t,1) \big( 1-\gr(t,0) \big) = q \; \gr(t,0) \big( 1- \gr(t,1) \big)\\
\partial_x \gr(t,0)  = \partial_x \gr(t,1)
\end{cases}
\end{eqnarray}
where $\gr(t,x)$ stands for the local density at the macroscopic time $t$ and position $x$.

\subsection{The case of a slowly varying field}
\label{sec: The case of a slowly varying field}

In contrast to the battery where the driving force is localized on  a single  bond, we will now consider  the case of a  smoothly varying macroscopic  field  $\mmu(x)$  along the ring.  The microscopic model is now the weakly asymmetric simple exclusion process of $M$ particles on a ring  of $N$ sites. 
A  particle  on site $i$ jumps to its neighboring site on its right at rate $p_i$ and to its neighboring site on its left at rate $q_i$ provided that the target site is empty. By weak asymmetry, we mean that $p_i$ and $q_i$ have the following scaling dependence on the system size $N$
 $$
 p_i = 1 + {1 \over N } \mmu \left( {i \over N} \right), \qquad
 q_i = 1 -  {1 \over N } \mmu \left( {i \over N} \right) \ .
 $$
When the integral of $\mmu(x)$ over the circle is non zero, the system will reach a non-equilibrium steady state for every fixed $N$.
 
\medskip 
 
The time evolution of the microscopic mean density is given by
\begin{equation}
\label{eq: evolution E}
\partial_\tau \mu_{N,\tau} (\eta_i) =
p_{i-1} \mu_{N,\tau} \big( \eta_{i-1} (1 - \eta_{i}) \big)
-q_{i} \mu_{N,\tau} \big( \eta_{i} (1 - \eta_{i-1}) \big)
-p_{i} \mu_{N,\tau} \big( \eta_{i} (1 - \eta_{i+1}) \big)
+q_{i+1} \mu_{N,\tau} \big( \eta_{i+1} (1 - \eta_{i}) \big)
\end{equation}
where $\mu_{N,\tau} (\cdot)$ is the expectation value wrt the time evolving measure.

Assuming that  for large $N$,   the density profile takes the following scaling form 
\begin{equation}
\mu_{N,\tau} \big( \eta_{i}  \big) =  \rho \left( {\tau \over N^2},{i \over N} \right) \, ,
\label{scal1}
\end{equation}
and the correlation functions scale as
\begin{equation}
\mu_{N,\tau} \big( \eta_{i}  \eta_{j} \big)- \mu_{N,\tau} \big(  \eta_{i} \big) \mu_{N,\tau} \big(  \eta_{j} \big)
 =  {1 \over  N} {\mathcal C} \left( {\tau \over N^2}, {i \over N},{j \over N} \right) \, , 
\label{scal2}
\end{equation}
one can show that for $x = i/N, t = \tau/N^2$
\begin{eqnarray*}
&& p_{i-1} \mu_{N,\tau} \big( \eta_{i-1} (1 - \eta_{i}) \big) -q_{i} \mu_{N,\tau} \big( \eta_{i} (1 - \eta_{i-1}) \big)\\
&& \qquad \qquad
\simeq {1 \over N} \left[  - \partial_x \rho(t,x) + 2 \mmu(x) \rho(t,x) (1-  \rho(t,x)) \right] + {1 \over N^2} W(t,x) \, ,
\end{eqnarray*}
where $W$ depends on the two-point correlation function $\cC$ defined in \eqref{scal2}.
One then gets from \eqref{eq: evolution E} the macroscopic evolution equation
\begin{equation}
\label{F-evol}
\partial_t \rho(t,x) =  \partial_x \left[  \partial_x \rho (t,x) -  2 \mmu(x)  \rho(t,x) (1 - \rho(t,x)) \right]  \, .
\end{equation}
This is the viscous Burgers equation on a ring $x \in [0,1]$.
A mathematical derivation of \eqref{F-evol} can be found in \cite{KL, HS}.

The macroscopic evolution of the battery  model (\ref{eq: heat equation}) can then be recovered by taking a large localized field 
$\mmu(x)$  at the origin with a given integral equal to $K$. 
According to (\ref{F-evol}) this leads, in the limit, to  a jump of density across the origin  
\begin{eqnarray*}
\log \frac{\grp}{1-\grp} - \log \frac{\grm}{1-\grm}
=
\int_{\grm}^{\grp} {d \rho \over \rho(1- \rho) } = 2 K \, .
\end{eqnarray*}
This expression is equivalent to \eqref{eq: constraint} with $2 K = \log (p/q)$. 
The effect of the battery should be understood as a force which maintains a fixed difference of chemical potentials across  the origin.

\subsection{The Zero Range Process and the moving battery model}
\label{sec: moving battery}

In section \ref{sec: Zero range process}, we will study the invariant measure of a ZRP driven by a battery. 
The microscopic dynamics is defined as follows.
At site $i$,  the occupation number $\eta_i$ can take any integer value and a particle at site $i$ performs a jump to the left or to the right with rate $g(\eta_i)$, where the function $g$ is increasing and $g(0) =0$. The driving force is modeled by modifying the jump rates between sites $1$ and $N$: a particle jumps from  $1$ to $N$ with rate  $q \; g(\eta_1)$ and from  $N$ to $1$ with rate $p \; g(\eta_N)$. The hydrodynamic limit for this model on $\bbZ$ was derived in \cite{LOV}.

\medskip

The ZRP driven by a battery is related to the SSEP on a ring driven by a "moving battery", i.e. a special tagged particle with asymmetric jump rates $p$ and $q$. One can think of this particle as being driven by an external field $E$ with $p/q = \exp(2  E)$ \cite{FGL, LOV}.
A current will be induced by the moving battery.
The mapping between the two models was exploited by \cite{LOV} to study the motion of this tagged particle in $\bbZ$. 
More precisely, the SSEP with a moving battery and with $M$ particles on the ring of length $N$ can be mapped onto a ZRP on a ring of length $M$ with $N-M$ particles and with the specific rates $g(n) = 1_{n \geq 1}$.
Let $1$ be the index of the asymmetric particle in the SSEP and $2, \dots, M$ the indices of the other particles, then the ZRP variable
$\eta_i$ stands for the number of empty spaces ahead of particle $i$.

\section{The macroscopic approach for the battery model}
\label{sec: Macroscopic approach}

A generic property of non-equilibrium systems, maintained in a steady state by contact with 
reservoirs  at unequal chemical potentials, is the presence of long range correlations \cite{Spohn,SC,DKS, DELO,BDGJL-cor}.
The steady state correlations  can be predicted from a macroscopic approach based on  "fluctuating hydrodynamics" \cite{Spohn, DELO, OS1}.
Alternatively, one can derive the density large deviation functional for the steady state and then recover the correlations by expanding the functional near the steady state \cite{BDGJL2, BDGJL-cor, D2, BDLvW}.
We will apply the latter approach to compute the correlations in the battery model introduced in section \ref{The battery model}.

\subsection{The hydrodynamic large deviations}

After rescaling the space by $N$ ($i = x N$) and the time by $N^2$ ($\tau = N^2 t$), the microscopic system can be described at the hydrodynamic scale by two macroscopic functions
the local density $\gr(t,x)$ and the local current of particles $q(t,x)$ which obey the conservation law
\begin{eqnarray}
\label{eq: conservation law}
\partial_t \gr(t,x) = - \partial_x q(t,x)  \, . 
\end{eqnarray}

The probability of observing a joint deviation of the current and the density in our microscopic system depends exponentially on the system size. 
More precisely, the probability of observing an atypical macroscopic trajectory $(\gr(t,x), q(t,x))_{0 \leq t \leq T}$ during the macroscopic time interval $[0,T]$ scales like
\begin{eqnarray*}
\bbP_{[0,T N^2]} \Big( (\gr(t,x), q(t,x)) \Big) 
\simeq 
\exp \left( - N {\hat \cF}_{[0,T]} \big( \gr,q \big) \right) \, , 
\end{eqnarray*}
where the large deviation functional is given by 
\begin{eqnarray}
\label{eq: mother}
&& {\hat \cF}_{[0,T]} \big( \gr,q \big) =
\int_0^T \; dt \; 
\int_0^1 \; dx  \frac{\big( q(t,x)  
+  \partial_x \gr (t,x) \big)^2}{2 \, \gs(\gr(t,x))} \, .
\end{eqnarray}
Here $\gs(u) = 2 u (1-u)$ and the density has to satisfy the conditions 
\begin{eqnarray}
\label{eq: bd conditions}
p \; \gr(t,1) (1-\gr (t,0)) = q \; \gr (t,0) (1- \gr (t,1) ), \qquad
\int_0^1 \, dx \; \gr(t,x) = \bar \gr\, .
\end{eqnarray}
If  $(\gr(t,x), q(t,x))$  does not satisfy \eqref{eq: conservation law} or if the conditions \eqref{eq: bd conditions} are not satisfied (for a set of times with non-zero Lebesgue measure) then the functional is infinite.
The expression \eqref{eq: mother} is a generalization of the large deviation functionals derived for open systems (see \cite{BDGJL5, D2, BD2007} for reviews).
The condition \eqref{eq: bd conditions} at the battery follows from \eqref{eq: constraint},
  as local equilibrium is still satisfied in the hydrodynamic large deviation regime.

\medskip

To compute the probability of observing an atypical density trajectory $(\gr(t,x))_{0 \leq t \leq T}$
\begin{eqnarray}
\label{eq: dev density}
\bbP_{[0,T N^2]} \Big( \gr(t,x) \Big) 
\simeq 
\exp \left( - N \cF_{[0,T]} \big( \gr \big) \right) \, , 
\quad {\rm with} \quad
\cF_{[0,T]} \big( \gr \big) = \inf_q {\hat \cF}_{[0,T]} \big( \gr,q \big) \, ,
\end{eqnarray}
one has to optimize ${\hat \cF}_{[0,T]} \big( \gr,q \big)$ over all  the currents which are compatible with the density \eqref{eq: conservation law}
and which can be written as
\begin{eqnarray}
\label{eq: conservation law 2}
q(t,x) = j(t) - \int_0^x \partial_t \gr(t,u) \, du  \, , 
\end{eqnarray}
where $j(t)$ is the current at 0.
Optimizing \eqref{eq: mother} over $j(t)$ implies
\begin{eqnarray}
\label{eq: H boundary 00}
\forall t>0, \qquad
0 =
\int_0^1 \; dx  \frac{ j(t) - \int_0^x \partial_t \gr(t,u) \, du  + \partial_x  \gr (t,x) }{ \gs(\gr(t,x))} \, .
\end{eqnarray}
We introduce an auxiliary function $H$ such that 
\begin{eqnarray}
\label{eq: current}
q(t,x) = - \partial_x  \gr (t,x)  + \gs(\gr(t,x)) \partial_x H(t,x) \, ,
\end{eqnarray}
and get from \eqref{eq: H boundary 00} 
\begin{eqnarray*}
0= \int_0^1 \, dx  \frac{ q(t,x) + \partial_x  \gr (t,x) }{ \gs(\gr(t,x))} 
= \int_0^1 dx  \partial_x  H (t,x) = H(t,1) - H(t,0) \, .
\end{eqnarray*}
This leads to a continuity condition on $H$ at the battery
\begin{eqnarray}
\label{eq: H continuous}
H(t,0)=H(t,1) \, .
\end{eqnarray}

For a trajectory $(\gr(t,x))_{0 \leq t \leq T}$, there exists $H$ satisfying \eqref{eq: current} and such that
\begin{eqnarray}
\label{eq: H, rho}
\partial_t \gr (t,x) =  \partial_x^2 \gr (t,x) - \partial_x \Big( \gs(\gr(t,x)) \partial_x H(t,x) \Big) \, ,
\end{eqnarray}
thus the functional \eqref{eq: dev density} reads 
\begin{eqnarray}
\label{eq: fonc density}
\cF_{[0,T]} \big( \gr \big) 
= \frac{1}{2} \int_0^T dt  \int_0^1 dx  \;  \gs(\gr(t,x)) \big(\partial_x H (t,x) \big)^2 \, .
\end{eqnarray}
Note that finding the optimal trajectories $(H,\gr)$ is equivalent to finding the optimal $(q,\gr)$.

\subsection{Steady state large deviations}
\label{subsec: Steady state large deviations}

Following \cite{BDGJL2}, the density large deviations for the steady state can be computed by using the hydrodynamic large deviations. For any smooth function $\gl (x)$ in $[0,1]$, one defines 
\begin{eqnarray}
\label{eq: G(lambda) 0}
\cG(\gl)
= \lim_{N \to \infty} \frac{1}{N} \log \left \bra 
\exp \left( \sum_{i=1}^N \gl  \left( \frac{i}{N} \right) \eta_i \right)  \right \ket  \, ,
\end{eqnarray}
which is given according to  \cite{BDGJL2} by 
\begin{eqnarray}
\label{eq: G(lambda)}
\cG(\gl) = \lim_{T \to \infty} \; \sup_{\gr}  
\left\{ \int_0^x \gl (x) \, \gr(0,x) - \cF_{[-T,0]} \big( \gr \big) \right\} \, ,
\end{eqnarray}
where $\cF_{[-T,0]}$ is defined by \eqref{eq: fonc density}.
The supremum ranges over all the density trajectories $\gr(t,x)$ in the macroscopic time interval $[-T,0]$ which are equal to the steady state $\bar \gr(x)$ at time $-T$.

\medskip

Given the function $\gl$ and a time $-T$, we look for the optimal density trajectory $\gr$ in the variational problem \eqref{eq: G(lambda)}. 
To determine this optimal trajectory, we consider a variation $\gr \to \gr + \gp$ and $H \to H + h$.
From \eqref{eq: H, rho}, $\gp$ and $h$ satisfy 
\begin{eqnarray}
\label{eq: h, phi}
\partial_t \gp  = \partial_x^2 \gp  - \partial_x \Big( \gs'(\gr) H^\prime \, \gp  + \gs(\gr) h^\prime \Big) \, .
\end{eqnarray}
To simplify notations, we have omitted the $(t,x)$ dependence and used the shorthand $H^\prime = \partial_x H(t,x)$.

The conditions \eqref{eq: bd conditions} become
\begin{eqnarray}
\label{eq: linear bd conditions}
\frac{\gp(t,1)}{ \gr(t,1) (1-\gr (t,1))} =
\frac{\gp(t,0)}{ \gr(t,0) (1-\gr (t,0))}, 
\qquad
\int_0^1 \, dx \, \gp(t,x) = 0 \, .
\end{eqnarray}
As the density is equal to $\bar \gr$ at time $-T$ then $\gp(-T,x) = 0$.

\medskip

For an optimal trajectory $\gr$,  the first order perturbation in $\gp$ in the variational problem \eqref{eq: G(lambda)} should be equal to 
0 so that 
\begin{eqnarray}
\label{eq: euler 0}
\int_0^1 dx \, \gl(x) \gp(0,x) - 
\int_{-T}^0 dt  \, \int_0^1 dx \; \left(
\frac{1}{2} \gs'(\gr) \big( H^\prime \big)^2 \; \gp +  \gs(\gr)  H^\prime \, h^\prime \right)  = 0 \, .
\end{eqnarray}
Integrating by parts and using \eqref{eq: h, phi}, one gets at time $t$
\begin{eqnarray*}
&& \int_0^1 dx \;  \gs(\gr) H^\prime h^\prime \\
&& \qquad = -  \int_0^1 dx \; \big( \partial_x \big( \gs(\gr) h^\prime \big) \big) H 
+ \left( H(t,1)  \gs(\gr(t,1)) h^\prime (t,1) -  H(t,0)  \gs(\gr(t,0)) h^\prime (t,0) \right) \\
&& \qquad =    \int_0^1 dx \;  \left( \partial_t \gp  - \partial_x^2 \gp 
+ \partial_x \Big( \gs'(\gr) H^\prime  \gp \Big)  \right) H  \\
&& \qquad \qquad + \left( H(t,1)  \gs(\gr(t,1)) h^\prime (t,1) -  H(t,0)  \gs(\gr(t,0)) h^\prime (t,0) \right) \, .
\end{eqnarray*}
Another integration by parts leads to 
\begin{eqnarray}
\label{eq: IPP 1}
&& \int_{-T}^0 dt    \int_0^1 dx \;  \gs(\gr)  H^\prime h^\prime =
- \int_{-T}^0 dt    \int_0^1 dx \;  \left[  \gp \partial_t H   -  \gp' \, H'  +  \gp \gs'(\gr) ( H^\prime )^2   \right]   \\
&& \qquad  +  \int_0^1 dx \; \gp (0,x) H (0,x) 
+ \int_{-T}^0 dt \,  H(t,1) j (t,1) -  H(t,0) j (t,0)   \, , \nonumber 
\end{eqnarray}
where $j(t,x) = - \gp^\prime (t,x)   + \gs^\prime (\gr(t,x)) H^\prime (t,x)  \gp(t,x) + \gs(\gr(t,x)) h^\prime (t,x)$ is the current variation
$q \to q +j$.
The current $j(t,x)$ is continuous at the battery 
(this can be checked by combining the condition $\int_0^1 \, dx \, \gp(t,x) = 0$ and \eqref{eq: h, phi}
which says that $\partial_t \gp = - j^\prime$).
As $H$ is also continuous at the battery \eqref{eq: H continuous}, the last boundary term in \eqref{eq: IPP 1} vanishes.
One has
\begin{eqnarray*}
&& \int_{-T}^0 dt    \int_0^1 dx \;   \gs(\gr)  H^\prime h^\prime  = 
\int_{-T}^0 dt    \int_0^1 dx \;  \left[  \gp \left( - \partial_t H   - \partial_x^2 H  -  \gs'(\gr) ( H^\prime )^2 \right)  \right]  \\
&& \qquad \qquad 
+  \int_0^1 dx \; \gp (0,x) H (0,x) + \int_{-T}^0 dt \,   \left[  H^\prime (t,1) \gp (t,1) - H^\prime (t,0) \gp (t,0) \right]  \, . \nonumber 
\end{eqnarray*}
Thus for any perturbation $\gp$, the condition \eqref{eq: euler 0} can be rewritten
\begin{eqnarray}
&&  \int_{-T}^0 dt    \int_0^1 dx \;  \left[  \gp \left(  \partial_t H  + \partial_x^2 H  + \frac{1}{2}  \gs'(\gr) ( H^\prime )^2 \right)  \right]
+ \int_0^1 dx \, \big( \gl(x) - H(0,x) \big) \gp(0,x)   \nonumber \\
&& \qquad 
- \int_{-T}^0 dt \,   \left[  \gs(\gr(t,1)) H^\prime (t,1)  - \gs(\gr(t,0)) H^\prime (t,0)  \right] \frac{\gp (t,0)}{\gs(\gr(t,0))}
= 0 \, ,
\label{eq: euler 1}
\end{eqnarray}
where we used the identity \eqref{eq: linear bd conditions}.

Combining \eqref{eq: euler 1} and \eqref{eq: H, rho}, the evolution equations for the optimal trajectory in \eqref{eq: G(lambda)} are
\begin{eqnarray}
\label{eq: optimal evolution}
\forall t \in [-T,0], \qquad
\begin{cases}
\partial_t \gr  = \partial_x^2 \gr - \partial_x \Big( \gs(\gr )  H^\prime  \Big) \\
\partial_t H  =  - \partial_x^2 H - \frac{1}{2} \gs'(\gr) \big( H^\prime \big)^2 
\end{cases}
\end{eqnarray}
with boundary conditions
\begin{eqnarray}
\label{eq: bd conditions time}
\forall x \in [0,1], \qquad \gr(-T,x) = \bar \gr (x), \qquad
H (0,x) = \lambda (x) \, .
\end{eqnarray}
We also have that at the battery, for any time $t<0$,
\begin{eqnarray}
\label{eq: bd conditions rho}
p \; \gr(t,1) (1-\gr (t,0)) = q \; \gr (t,0) (1- \gr (t,1) ), \qquad
\gr^\prime (t,0) = \gr^\prime (t,1) \, ,
\end{eqnarray}
\begin{eqnarray}
\label{eq: additional bd conditions}
H(t,0)=H(t,1), 
\qquad 
\gr(t,1) (1-\gr (t,1)) H' (t,1) = \gr(t,0) (1-\gr (t,0)) H' (t,0) \, .
\end{eqnarray}
The boundary conditions follow from \eqref{eq: bd conditions}, \eqref{eq: H continuous} and \eqref{eq: euler 1}. The continuity of $\gr^\prime$ 
in \eqref{eq: bd conditions rho}  is a consequence of the continuity of the current \eqref{eq: current} 
 at the battery and of \eqref{eq: additional bd conditions}. We stress the fact that $H (0,x) =\gl(x)$ in \eqref{eq: bd conditions time} may not satisfy the boundary conditions \eqref{eq: additional bd conditions}.
 
The bulk evolution equations \eqref{eq: optimal evolution} have already  been derived in previous works \cite{BDGJL2, BDLvW, DG}, but the
boundary conditions \eqref{eq: bd conditions rho}, \eqref{eq: additional bd conditions} are specific to the battery model.

\subsection{Linearized equations}

The  two-point correlations in the steady state can be obtained by taking the second derivative of $\cG(\gl)$ at $\gl =0$ (see \eqref{eq: G(lambda) 0})
 \cite{BDGJL2, D2}. This can be understood as follows. For a given $N$, one has at the second order in $\gl$ 
\begin{eqnarray}
\label{eq: dvp discret G}
&& \frac{1}{N} \log \left \bra \exp \left( \sum_{i=1}^N \gl  \left( \frac{i}{N} \right) \eta_i \right)  \right \ket  \\
&& \quad = 
\frac{1}{N} \sum_{i=1}^N \left(  \gl  \left( \frac{i}{N} \right) \left \bra  \eta_i  \right \ket  
+ \frac{1}{2} \gl  \left( \frac{i}{N} \right)^2  \bra \eta_i \ket (1- \bra \eta_i  \ket) \right) 
+ \frac{1}{2 N} \sum_{i \not = j}^N  \gl  \left( \frac{i}{N} \right) \gl  \left( \frac{j}{N} \right) \left \bra  \eta_i ; \eta_j  \right \ket
+ o(\gl^2)  \, , \nonumber
\end{eqnarray}
where $\bra  \eta_i ; \eta_j  \ket = \bra  \eta_i \eta_j  \ket - \bra  \eta_i  \ket \, \bra  \eta_j  \ket$ is the connected two-point correlation function.
The above expansion amounts to perturbing the system around the steady state $\bar \gr(x)$ \eqref{eq: bar rho}.

For $N$ large, one expects from \eqref{eq: G(lambda) 0} and \eqref{eq: dvp discret G} that 
 \begin{eqnarray}
\label{eq: G(lambda) 1}
\cG(\gl)
=  \int_0^1 dx \;   \bar \gr(x) \gl (x)+ \bar \gr(x) (1-\bar \gr(x))  \frac{\gl (x)^2}{2} +  \int_0^1 dx \int_x^1  dy \, C(x,y) \gl (x) \gl (y)  \, ,
\end{eqnarray}
where $C(x,y)$ is the macroscopic non-equilibrium two-point correlation function
 \begin{eqnarray}
\label{eq: micro-macro}
\bra  \eta_i ; \eta_j  \ket = \frac{1}{N} C \left( \frac{i}{N} , \frac{j}{N} \right)  \, ,
\end{eqnarray}
at the leading order in $N$.
$C(x,y)$ will be computed in \eqref{eq: C(x,w)}.
We stress the fact that the equality \eqref{eq: G(lambda) 1} rests on the assumption that one can interchange the limits
$N \to \infty$ and $\gl \to 0$.
 
Our goal now is to compute  $\cG(\gl)$ to the second order in $\gl$ from the variational problem  \eqref{eq: G(lambda)}
and to deduce from it $C(x,y)$ \eqref{eq: G(lambda) 1}.
To any $\gl$, one can associate an optimal density  $\gr_\gl$ which minimizes the variational problem  \eqref{eq: G(lambda)}. 
By construction the density at time $0$ satisfies $\gr_\gl (0,x) = \frac{\delta  \cG(\gl)}{\delta \gl}$ so that 
expanding $\gr_\gl (0,x)$ to the first order in $\gl$ will lead to the second derivative of $\cG(\gl)$. 
One can see from \eqref{eq: G(lambda) 1} that
\begin{eqnarray}
\label{eq: expansion}
\gr_\gl (0,x) = \bar \gr(x) + \bar \gr(x) (1-\bar \gr(x)) \gl (x) + \int_0^1 dy \, C(x,y)  \gl (y) \, .
\end{eqnarray}


To determine $\gr_\gl$ for small $\gl$, we linearize the optimal evolution \eqref{eq: optimal evolution} around the steady state profile $\bar \gr$
\eqref{eq: bar rho}.
For the sake of notation, we drop the subscript $\lambda$ in $\gr_\gl$  and decompose the trajectory as $\gr_\gl = \bar \gr + \gp$ and $H = h$, 
where $\gp$ and $h$ are small. 
At the first order, \eqref{eq: optimal evolution} becomes
\begin{eqnarray}
\label{eq: linear evolution}
\forall t \in [-T,0], \qquad
\begin{cases}
\partial_t \gp  = \partial_x^2 \gp - \partial_x \Big( \gs(\bar \gr )  h^\prime  \Big) \\
\partial_t h  =  - \partial_x^2 h 
\end{cases}
\end{eqnarray}
and the boundary conditions \eqref{eq: bd conditions time}, \eqref{eq: bd conditions rho}, \eqref{eq: additional bd conditions}
lead to 
\begin{eqnarray}
\label{eq: bd conditions time linear}
\forall x \in [0,1], \qquad \gp(-T,x) = 0, \qquad
h (0,x) = \lambda (x) \, ,
\end{eqnarray}
and
\begin{eqnarray}
\label{eq: bd conditions phi}
\frac{\gp(t,1)}{\gs \big( \grm \big) } =
\frac{\gp(t,0)}{\gs \big( \grp \big) }, 
\qquad \gp^\prime (t,0) = \gp^\prime (t,1) \, ,
\end{eqnarray}
\begin{eqnarray}
\label{eq: bd conditions h}
h (t,1) =  h (t,0) , \qquad 
\gs \big( \grm \big) \, h' (t,1) = \gs \big( \grp \big) \, h' (t,0) \, .
\end{eqnarray}
Note that \eqref{eq: bd conditions phi} has been obtained by linearizing \eqref{eq: bd conditions rho} and that 
$h$ is defined up to a constant which does not influence the evolution of the density.

\bigskip

The coupled equations \eqref{eq: linear evolution} can be solved by integrating the heat equation.
We first start with the second equation.  
The Green's function $G^2$ associated to the Laplacian with boundary conditions \eqref{eq: bd conditions h}
satisfies for any $y$ in $(0,1)$
\begin{eqnarray}
\partial_t G^2_t (x,y) = \partial_x^2 G^2_t (x,y) ,
\quad {\rm with} \qquad
G^2_t (1,y) =  G^2_t (0,y), \quad  a \partial_x G^2_t (1,y) = \partial_x G^2_t (0,y),
\label{eq: bords  G2}
\end{eqnarray}
where we set 
\begin{eqnarray}
\label{eq: a}
a = \frac{ \gs( \grm ) }{\gs( \grp )} \, .
\end{eqnarray}
The function $h$ follows the backward  heat equation on $[-T,0]$ with final data $h(0,y) = \gl(y)$ at time $0$.
Therefore, we can write
\begin{eqnarray*}
\forall t \in [-T,0], \qquad
h(t,x)   =   \int_0^1 G^2_{-t} (x,y)  \gl(y) \, dy  \, .
\end{eqnarray*}

We turn now to the evolution of $\gp$ in \eqref{eq: linear evolution} which follows the heat equation with a source term depending on $h$.
The Green's function $G^1$ associated to the Laplacian with boundary conditions \eqref{eq: bd conditions phi} satisfies for $y$ in $(0,1)$
\begin{eqnarray}
\partial_t G^1_t (x,y) = \partial_x^2 G^1_t (x,y) ,
\quad {\rm with} \qquad
G^1_t (1,y) = a  G^1_t (0,y), 
\quad   \partial_x G^1_t (1,y) = \partial_x G^1_t (0,y) \, .
\label{eq: bords G1}
\end{eqnarray}
Given $\gp(-T,y) = 0$ at time $-T$ 
\begin{eqnarray*}
\forall t \in [-T,0], \qquad
\gp (t,x)  = - \int_{-T}^t \, ds \int_0^1 \, dy  \; G^1_{t-s} (x,y) \partial_y \Big( \gs(\bar \gr (y) ) \partial_y h (s,y)  \Big) \, .
\end{eqnarray*}

\medskip

One can check directly that the Green's function solution of  \eqref{eq: bords G1} is given by 
\begin{eqnarray}
\label{eq: G1 0}
G^1_t (x,y) 
&=&  \frac{2}{a+1}  \ga(x)  + \frac{4}{a+1} \sum_{k = 1}^\infty   \exp \big(  - (2 k \pi)^2 t \big)
\Big(   \ga(x)  \cos ( 2 \pi k x) \cos ( 2 \pi k y)  +\\
&&  \ga(1-y)  \sin ( 2 \pi k x) \sin ( 2 \pi k y)  + 4 k \pi (1-a) t \sin ( 2 \pi k x) \cos ( 2 \pi k y)   \Big)  \, , \nonumber
\end{eqnarray}
with $\ga(x) = a x + 1 -x$ and $a = \frac{ \gs( \grm ) }{\gs( \grp )}$. 
One can also check that the Green's function $G^2$ which solves \eqref{eq: bords  G2}  is given by
\begin{eqnarray}
G^2_t (x,y) =  G^1_t (y,x) \, .
\label{eq: G1,G2}
\end{eqnarray}
In appendix I,  the construction of the Green's functions is explained.

\bigskip

We are going to relate $\gp(0,x)$ and $\gl(x)$. 
The density fluctuation at the final time reads
\begin{eqnarray*}
\gp (0,x)  =  -
\int_{-T}^0 \, ds \int_0^1  dy \int_0^1  dw \,  G^1_{-s} (x,y) \partial_y \Big( \gs(\bar \gr (y) ) \partial_y  G^2_{-s} (y,w)  \Big) \gl(w) \, ,
\end{eqnarray*}
which is equivalent to
\begin{eqnarray}
\gp (0,x)  =  - \int_0^1  dw \left[ \int_0^T \, ds \int_0^1  dy  \,
 G^1_s (x,y) \partial_y \Big( \gs(\bar \gr (y) ) \partial_y  G^2_s (y,w)  \Big) \right] \;  \gl(w) \, .
\label{eq: phi, h}
\end{eqnarray}
In appendix II, it is shown that \eqref{eq: phi, h} can be rewritten as
\begin{eqnarray}
\label{eq: contrainte } 
\gp (0,x) = \frac{1}{2} \gs(\bar \gr (x) )  \gl (x) + \int_0^1 dy \, C(x,y)  \gl (y) \, , 
\end{eqnarray}
where the long range correlations are given by  
\begin{eqnarray}
\label{eq: C(x,w)} 
 C(x,y) &=&  
 - \frac{2}{(a+1)^2}  \left( \frac{\gs(\grp ) + \gs(\grm)}{2} +  \frac{\cJ^2}{3} \right)  (ax + 1-x) (ay + 1-y) \\
&& \qquad -  2 \cJ^2 \, \int_0^\infty ds  \int_0^1 dz \,   G^1_s (x,z)  G^1_s (y,z) -   G^1_s (x,0)  G^1_s (y,0) \, . \nonumber
\end{eqnarray}
Since $\gr_\gl (0,x) = \bar \gr(x) +\gp(0,x)$, this determines the correlation function in \eqref{eq: expansion}.

\section{The macroscopic approach for the slowly varying field}
\label{sec: The macroscopic approach for the slowly varying field}

In section \ref{sec: The case of a slowly varying field}, a microscopic dynamics driven by a weak field $\mmu(x)$ was introduced.
We consider now a general diffusive system at a macroscopic level and write down the two-point correlation function which generalizes (\ref{eq: C(x,w)}) to an arbitrary weak field $\mmu(x)$.
In the limit where $\mmu(x)$ becomes a $\gd$-function, we will see that 
one can  recover the correlations of the battery model  (\ref{eq: C(x,w)}).

The main advantage of the calculations of this section is that the effect of the battery is smoothened over the whole system 
through the function $\mmu(x)$, so that there are no longer boundary conditions such as (\ref{eq: bd conditions rho}) or (\ref{eq: additional bd conditions})  at the battery and the functions $H(t,x)$ and $\rho(t,x)$ become  simply periodic functions of the space variable $x$.
One difficulty is that for a generic $\mmu(x)$ the steady state profile $ \bar \rho(x)$ and the Green's functions $G^1 $ and $G^2$ are not known explicitly  and we have to make a few assumptions on the convergence of hydrodynamic equation \eqref{eq: hydro E}
 to the steady state $\bar \rho(x)$ or on the long time behavior of the Green's functions \eqref{B-sol}.

\subsection{The variational problem}

Our starting point is  that the large deviation functional  (\ref{eq: mother}) of a macrosopic trajectory  $(\gr(t,x), q(t,x))_{0 \leq t \leq T}$
  which, for a general diffusive system is characterized by two functions $D(\rho)$ and $\sigma(\rho)$, 
  takes the following form  \cite{BDGJL2006, BD2007}   in presence of a driving field $\mmu$
\begin{eqnarray*}
\label{eq: mother2}
&& {\hat {\mathcal F}^\mmu }_{[0,T]} \big( \gr,q \big) =
\int_0^T  dt  
\int_0^1  dx  \frac{\Big( q(t,x)  
+ D(\gr(t,x))\partial_x \gr (t,x)  - \mmu(x) \sigma(\gr(t,x))\Big)^2}{2 \, \gs(\gr(t,x))} \, ,
\end{eqnarray*}
The SSEP with a weak field introduced in section \ref{sec: The case of a slowly varying field} corresponds 
to the functions $D(\rho) = 1$ and $\sigma(\rho) = 2 \gr (1-\gr)$.
Following exactly the same steps as in subsection 3.1, one gets, by optimizing over the current $j(t),$ defined in 
(\ref{eq: conservation law 2}),  that (\ref{eq: current})  becomes
\begin{eqnarray}
\label{eq: current1}
q(t,x) = -  D(\gr(t,x)) \; \partial_x  \gr (t,x)  + \gs(\gr(t,x))  [\mmu(x)+ \partial_x H(t,x)]  \, ,
\end{eqnarray}
where $H(t,x)$ is  periodic in space  ($H(t,x)=H(t,x+1)$). This implies that  the  time dependent density profile $\gr$  is related to $H$  by
\begin{eqnarray}
  \label{eq: H, rho1}
\partial_t \gr (t,x) =  \partial_x \Big(D(\gr(t,x))\;  \partial_x \gr (t,x) \Big) - \partial_x \Big( \gs(\gr(t,x)) \; [ \mmu(x) +  \partial_x H(t,x)] \Big) \, ,
\end{eqnarray}
(instead of (\ref{eq: H, rho}))
and the density large deviation functional  is given by  
\begin{eqnarray}
\label{eq: fonc density1}
{\mathcal F}^\mmu_{[0,T]} \big( \gr \big) 
= \frac{1}{2} \int_0^T dt  \int_0^1 dx  \;  \gs(\gr(t,x)) \big(\partial_x H (t,x) \big)^2 \, ,
\end{eqnarray}
as in (\ref{eq: fonc density}).

\medskip

To evaluate  $\mathcal G$ defined  in (\ref{eq: G(lambda) 0}), (\ref{eq: G(lambda)}) we proceed as in subsection \ref{subsec: Steady state large deviations} and determine the optimal dynamical fluctuation starting from the steady state $\bar \gr(x)$ given as the solution of 
\begin{equation}
\label{steady-state}
-D \big( \bar \gr(x) \big) \bar \gr' (x) + \mmu(x) \gs \big( \bar \gr(x) \big) =  \cJ \, ,
\end{equation}
where $\cJ$ is the steady state current.
In Appendix III, we check that for regular coefficients  $\sigma(\rho),D(\rho) >0$ and $\mmu(x)$ considered here \eqref{steady-state} has a unique solution.
We shall further assume that the hydrodynamic evolution 
\begin{eqnarray}
\label{eq: hydro E}
\partial_t \gr  = \partial_x \Big(D(\rho)   \partial_x\gr  - \mmu(x) \gs(\gr)   \Big) \, ,
\end{eqnarray}
converges to  the steady state profile $\bar \rho(x)$ \eqref{steady-state}.

\medskip

As in subsection \ref{subsec: Steady state large deviations}, one considers a small variation
 $\rho \to \rho+ \varphi, H \to H+h$ of  $\rho(t,x)$ and of $H(t,x)$
and  one gets from (\ref{eq: H, rho1}) that (\ref{eq: h, phi}) becomes
\begin{eqnarray}
\label{eq: h, phi1}
\partial_t \gp  = \partial_x^2 \Big(D(\gr) \gp \Big)  - \partial_x \Big(\mmu(x) \gs'(\rho) \gp +  \gs'(\gr) H^\prime \, \gp  + \gs(\gr) h^\prime \Big) \, .
\end{eqnarray}
Then following exactly the same steps as in subsection \ref{subsec: Steady state large deviations}, 
one ends up with (\ref{eq: optimal evolution}) replaced by 
\begin{eqnarray}
\label{eq: optimal evolution1}
&& \partial_t \gr  = \partial_x \Big(D(\rho)   \partial_x\gr  - \mmu(x) \gs(\gr) - \gs(\gr )    \partial_x H  \Big) \\
\label{eq: optimal evolution2}
&& \partial_t H  =  - D(\rho) \partial_x^2 H - \mmu(x) \gs'(\gr)  \partial_x H-\frac{1}{2} \gs'(\gr) \big( \partial_x  H \big)^2 
\end{eqnarray}
with the boundary conditions at times $-T$ and 0 as in (\ref{eq: bd conditions time}) 
\begin{eqnarray}
\label{eq: bd conditions time1}
\forall x \in [0,1], \qquad \gr(-T,x) = \bar \gr (x), \qquad
H (0,x) = \lambda (x) \, .
\end{eqnarray}
The spatial boundary conditions (\ref{eq: bd conditions rho}), (\ref{eq: additional bd conditions}) are now
replaced by the requirement that $H(t,x)$ and $\rho(t,x)$ are smooth periodic functions of the space variable $x$ (except possibly at time $t=0$).

\subsection{Small variations of the density}

In order to determine the correlation function $C_\mmu (x,y)$ as in (\ref{eq: expansion}) we need to solve the above equations 
(\ref{eq: optimal evolution1} -- \ref{eq: bd conditions time1}) to first order in $\lambda(x)$, that is to first order in
$\gp= \rho-\bar \gr$ and $h= H$. 
Linearizing  (\ref{eq: optimal evolution1}), (\ref{eq: optimal evolution2}), one gets
\begin{eqnarray}
\label{eq: optimal evolution3}
&& \partial_t \gp  = \partial_x \Big( \partial_x(D( \bar \rho)   \gp)  - \mmu(x) \gs'(\bar \gr) \gp - \gs(\bar \gr)    \partial_x h \Big) \\
\label{eq: optimal evolution4}
&& \partial_t h  =  - D(\bar \rho) \partial_x^2 h - \mmu(x) \gs(\bar \rho)'  \partial_x h
\end{eqnarray}
As in section \ref{sec: Macroscopic approach}, 
these equations can be solved in two steps. 
As the equation for $h$ does not involve $\gp$, it is determined in terms of the Green's function $G^2$ as
\begin{eqnarray}
\forall t \in [-T,0], \qquad
h(t,x)   =   \int_0^1 G^2_{-t} (x,y)  \lambda(y) \, dy  \, .
\label{h-sol}
\end{eqnarray}
where $G^2$ is solution of
\begin{eqnarray}
\partial_t G^2_t (x,y) 
= D(\bar \gr(x)) \; 
\partial_x^2 G^2_t (x,y) 
 +   \mmu(x)  \sigma' (\bar \gr(x))  \; \partial_x G^2_t (x,y)   \ ,
 \label{eq: nouveau G2}
\end{eqnarray}
with $G^2_0 (x,y) = \delta_{x=y}$.
To solve for $\gp$, it is convenient to write $\gp$ as
\begin{equation}
\gp(t,x)= {\sigma(\bar{\rho}(x)) \over 2 D(\bar{\rho}(x))} h(t,x) + \psi(t,x) \, .
\label{psi-def}
\end{equation}
 The evolution equation of $\psi$  can be determined from 
(\ref{eq: optimal evolution2})
\begin{equation}
\partial_t \psi = \partial_x^2 [ D \psi ] -  \partial_x [\mmu \sigma' \psi]  +  h \partial_x  \left[{\sigma' \bar \rho' \over 2}  -{\mmu \sigma \sigma' \over 2 D} \right]
\nonumber
\end{equation}
where the functions $D,\gs,\gs'$ are evaluated at density $\bar \rho(x)$ and
$\mmu$ at position $x$.
Using the fact that the steady state  profile $\bar \rho(x)$
satisfies (\ref{steady-state}), 
this becomes 
\begin{equation}
\partial_t \psi =\partial_x^2 [ D \psi ] -  \partial_x  [\mmu \sigma' \psi]  -  \cJ  h \partial_x \left[{ \sigma' \over 2 D} \right] \, .
\label{psi-eq}
\end{equation}

We introduce the Green's function $G_t^1$ as the  solution of
\begin{eqnarray}
\partial_t G^1_t (x,y) 
=
\partial_x^2 
 \Big( D(\bar \gr(x)) \; 
G^1_t  (x,y) \Big) -  \partial_x \Big( \mmu(x)  \sigma' (\bar \gr(x))  
G^1_t (x,y)  
\Big) \  ,
\label{eq: nouveau G1}
\end{eqnarray}
with $G^1_0 (x,y) = \delta_{x=y}$.
Then one can solve  (\ref{psi-eq}) 
\begin{eqnarray}
\label{gp-sol}
\gp(x,0)= {\sigma(\bar{\rho}(x)) \over 2 D(\bar{\rho}(x))} \lambda(x) - 
\cJ \int_{0}^T  dt \int d z \int dy  \; G_t^1(x,z)    \partial_z  \Big( { \sigma'(\bar \rho(z)) \over 2 D(\bar \rho(z))} \Big) G_t^2(z,y)  \lambda(y) \\
- \int dz  \int dy \; G_T^1(x,z) {\sigma(\bar{\rho}(z)) \over 2 D(\bar{\rho}(z))} G_T^2(z,y) \lambda(y) \, .
\nonumber
\end{eqnarray}

We assume for all $x$ the convergence when $T \to \infty$
\begin{equation}
\label{B-sol}
G_T^1(x,z) =  G_T^2(z,x) \to B(x) \, .
\end{equation}
For the battery model $B(x)$ has the explicit form $\frac{2  \ga(x)}{a+1}$ (see \eqref{eq: G1 0}).
Taking the limit $T \to \infty$ in (\ref{gp-sol}) gives the correlation functions
\begin{equation*}
C_\mmu (x,y) =- B(x) B(y) \int dz  \;  {\sigma(\bar{\rho}(z)) \over 2 D(\bar{\rho}(z))}
 -  \cJ  \int_{0}^\infty  dt \int d z   \; G_t^1(x,z)    \partial_z  \Big( { \sigma'(\bar \rho(z)) \over 2 D(\bar \rho(z))} \Big) G_t^2(z,y)  \, .
\end{equation*}
As for the battery model, the Green's functions satisfy at any time $t$  the symmetry property \eqref{eq: G1,G2} (see the proof below)
\begin{eqnarray}
\label{eq: symm cas general}
G_t^2 (x,y)= G_t^1 (y,x) \, .
\end{eqnarray}
Thus we finally get the macroscopic expression for the two point correlation functions with a weak field
\begin{equation}
\label{C-sol}
C_\mmu (x,y) =- B(x) B(y) \int dz    {\sigma(\bar{\rho}(z)) \over 2 D(\bar{\rho}(z))}
 -  \cJ  \int_{0}^\infty  dt \int d z   \; G_t^1(x,z) G_t^1 (y,z)    \partial_z \Big( { \sigma'(\bar \rho(z)) \over 2 D(\bar \rho(z))} \Big)   \, .
\end{equation}

\bigskip

We sketch a proof of the symmetry \eqref{eq: symm cas general} which follows from the fact that the operators associated to the evolutions
\eqref{eq: nouveau G2}, \eqref{eq: nouveau G1} are adjoint. A solution $f(t,x)$ of
$$
\partial_t f(t,x)  = D(\bar{\rho}(x)) \partial_x^2 f(t,x)
 + \mmu(x) \sigma'(\bar{\rho}(x))  \partial_x f(t,x)  
$$
with initial condition $f(0,x)$, is, by definition of $G^2$  \eqref{eq: nouveau G2}, equal to
$$f(t,x) = \int dy G_t^2 (x,y) f(0,y) \, . $$
Now using the fact that
$$f(t+dt,x) = \int dy G_{t+dt}^2 (x,y) f(0,y) = \int dy G_t^2 (x,y) f(dt,y) \, ,$$
one  can show after an integration by parts that
$$
\partial_t G^2  = \partial_y^2 [D(\bar{\rho}(y)) G^2 ] - \partial_y [ \mmu(y) \sigma'(\bar{\rho}(y))  G^2] \, . 
$$
Therefore $G_t^2(x,y)$ evolves according to the same equation as $G_t^1(y,x)$ \eqref{eq: nouveau G1} and they coincide at time $t=0$ ($G_0^2(x,y)=G_0^1(y,x)=\delta_{x=y}$). Thus one concludes that they are identical.

\subsection{The case of a battery}

To represent our model of a battery as a limit of the smooth continuum problem discussed above, we consider the field $\mmu(x)$ 
localized in a certain  region of size $\Delta x$ around the origin  with $\mmu(x)=0 $ outside the  region 
$[1- \frac{\Delta x}{2},1] \cup [0, \frac{\Delta x}{2}]$ (where $1$ is identified to $0$ on the ring). 
Then one  should take  the difference of the chemical potentials 
$\mu(\frac{\Delta x}{2})$, $\mu(1-\frac{\Delta x}{2})$ at the edge of the battery to be fixed.  
Using the relation between the chemical potential and the density of systems in local equilibrium, one has
\begin{equation}
\mu \left( \frac{\Delta x}{2} \right) - \mu \left( 1-\frac{\Delta x}{2} \right)
=
\int_{\rho(t,1-\frac{\Delta x}{2})}^{\rho(t,\frac{\Delta x}{2})} {2 D(\rho ) \over \sigma(\rho)} d \rho=  K' \, ,
\label{bat-Delta-x}
\end{equation}
where the constant $K'$ is a characteristic 
of the battery which  remains fixed  as $\Delta x \to 0$.

For $K'>0$,  the picture is that for $\Delta x$ small, the density has an abrupt increase for $x$ in  the region where the field is localized 
and a slow decay for $\frac{\Delta x}{2} < x < 1- \frac{\Delta x}{2}$.
In the limit $\Delta x \to 0$, the steady state profile $\bar \rho(x)$  satisfies $- D(\bar \rho(x)) \bar \rho'(x) = \cJ$ for $0<x<1$ with a jump from 
$\grm = \bar\rho(1)$ to $\grp = \bar \rho(0)$. The current $\cJ$ and the densities $\grm$ and $\grp$ at the edges of the battery then satisfy
$$ 
\int_{\grm}^{\grp} {2 D(\rho) \over \sigma(\rho)} d \rho = K',
 \ \ \ \ \ \  \int_{0}^{1}  \bar \rho(x) dx = \bar \rho,
  \ \ \ \ \int_{\grm}^{\grp}  D(\rho)  d \rho = \cJ \, .
$$
If we assume that in the limit $\Delta x \to 0$, the Green's function $G_t^1(x,z)$  converges to $\hat G_t^1(x,z)$
 then expression (\ref{C-sol}) becomes 
\begin{eqnarray}
\hat C (x,y) &=& - B(x) B(y) \int dz    {\sigma(\bar{\rho}(z)) \over 2 D(\bar{\rho}(z))}
 - \cJ \int_{0}^\infty  dt \int_0^1 d z   \; \hat G_t^1(x,z)   \partial_z \Big( { \sigma'(\bar \rho(z)) \over 2 D(\bar \rho(z))} \Big) \hat G_t^1(y,z) \nonumber \\ 
 && \qquad   - \cJ \int_{0}^\infty  dt   \left[  { \sigma'(\grp)  \over 2 D(\grp)} -   { \sigma'(\grm)  \over 2 D(\grm)}  \right] 
\hat G_t^1(x,0) \; \hat G_t^1(y,0) \, .
\label{C-sol1}
\end{eqnarray}
It is then easy to check that this expression of the pair correlation functions, for a general diffusive system, reduces to the one of the battery model (\ref{eq: C(x,w)}) when one takes $D(\rho)=1$ and $\sigma(\rho)=2 \rho(1- \rho)$.

\section{Comparison with the open system and with numerical simulations}
\label{sec: Comparison with the open system and with numerical simulations}

\subsection{Comparison with the open system}

For open diffusive non-equilibrium systems, the steady state has long range correlations which scale like $1/N$ for a one-dimensional chain of length $N$. 
The two-point correlation function of the SSEP in contact with reservoirs at densities $\gr_a$, $\gr_b$ is given by
\cite{Spohn, DLS}
\begin{eqnarray*}
1 \leq i < j \leq N, \qquad 
\bra \eta_i \, ; \eta_j \ket_{\rm{open}} = \frac{1}{N} C_{\rm{open}} \left( \frac{i}{N} , \frac{j}{N}  \right) \, ,
\end{eqnarray*}
with
\begin{eqnarray}
\label{eq: C open} 
x< y, \qquad
C_{\rm{open}}(x,y) = - \cJ_{\rm{open}}^2 \;  x (1-y) \, ,
\end{eqnarray}
where $\cJ_{\rm{open}} = \gr_a -\gr_b$ stands for the steady state current. 

The SSEP on a ring driven by a battery can be interpreted as the canonical non-equilibrium counterpart of the SSEP driven by reservoirs. By tuning the mean density $\bar \gr$ and the asymmetry $(p,q)$, the boundary conditions at the battery 
$\gr^+,\gr^-$ \eqref{eq: gr+ gr-} can be chosen to be equal to $\gr_a,\gr_b$ so that both systems have the same steady state density profile and the same steady state currents $\cJ = \cJ_{\rm{open}}$. 
In both models local equilibrium holds and implies that at the leading order both steady state statistics are given locally by a product Bernoulli measure with the same density. In this sense, equivalence  of ensembles is also satisfied in non-equilibrium. 
In order to understand more precisely the interplay between the non-equilibrium correlations, the canonical constraint and the driving mechanism we are going to compare the correlation functions \eqref{eq: C(x,w)} and \eqref{eq: C open} (as well as \eqref{C-sol}).

\medskip

Expression \eqref{eq: C open} is related to the inverse of the Laplacian $\gD^{\rm{Dir}}$ in $[0,1]$ with Dirichlet boundary conditions.
Let $G^{\rm{Dir}}_s = \exp ( - s \gD^{\rm{Dir}})$ be the corresponding Green's function, then we know from
\cite{Spohn} that 
\begin{eqnarray*}
C_{\rm{open}}(x,y)  = -  \cJ^2 \big( \gD^{\rm{Dir}} \big)^{-1} (x,y) \, .
\end{eqnarray*}
This can be rewritten as
\begin{eqnarray*}
C_{\rm{open}}(x,y) 
= - 2 \cJ^2 \int_0^\infty ds \,  G^{\rm{Dir}}_{2s} (x,y)  
=  - 2 \cJ^2 \int_0^\infty ds  \int_0^1 dz \,   G^{\rm{Dir}}_s (x,z)  G^{\rm{Dir}}_s (y,z) \, ,
\end{eqnarray*}
where $\cJ = \cJ_{\rm{open}}$.
This is reminiscent of the expressions derived for the battery model  \eqref{eq: C(x,w)}  and the weak field \eqref{C-sol} with two notable differences. The Green's functions are not associated to the same operator (they differ by their boundary conditions or the bulk field). 
In \eqref{eq: C(x,w)} the contribution of the battery is
given by the term $G^1_s (x,0)  G^1_s (y,0)$ which is absent in the open case.
For a general field the current $\cJ$ is replaced by $\partial_z  \big( { \sigma'(\bar \rho(z)) \over 2 D(\bar \rho(z))} \big)$ in \eqref{C-sol}.

\medskip

The response to a small drive of the open and close systems is different.
For a small current $\cJ$,  when the mean density $\bar \gr \not = 1/2$ then the long range correlation \eqref{eq: C(x,w)}
 is of order $\cJ$. Instead  \eqref{eq: C open}  scales like $\cJ^2$.
To see this, we first note that 
in the battery model the mean density is linear and $\int_0^1 dx \, \bar \gr(x) = \frac{\grm + \grp}{2} = \bar \gr$. Thus
\begin{eqnarray}
\label{eq: relation 1 0}
\grm   = \bar \gr- \frac{\cJ}{2}, \qquad 
\grp = \bar \gr +  \frac{\cJ}{2} \, ,
\end{eqnarray}
where the mean current is given by $\cJ =  \grp  -  \grm$.
Expanding \eqref{eq: C(x,w)} to the first order in $\cJ$ leads to 
\begin{eqnarray}
\label{eq: C(x,w) first order} 
 C(x,y) = - \frac{\gs(\bar \gr)}{2} \left( 1 + \frac{\gs'(\bar \gr)}{\gs(\bar \gr)} \cJ (1 - x - y) \right) \, ,
\end{eqnarray}
where we used the fact that $a = 1 - \frac{\gs'(\bar \gr)}{\gs(\bar \gr)} \cJ$  to first order in $\cJ$.
We note that this coincides with the first order expansion of the correlations \eqref{eq: correlation independant}
for the canonical measure associated to the product measure. 
At the order $\cJ^2$ the non equilibrium contributions appear and \eqref{eq: C(x,w)}  no longer matches with
\eqref{eq: correlation independant}.

\subsection{The battery model at density $1/2$}

An explicit expression of the two-point correlations can be obtained by plugging in \eqref{eq: C(x,w)} the expression \eqref{eq: G1 0} of $G^1_s$. 
For general mean density, this expression is rather complicated. However, several simplifications occur when the mean density is equal to $1/2$ and in the rest of this section we will focus on this case.

At mean density equal to $1/2$, the relation \eqref{eq: relation 1 0} implies that the densities at the battery satisfy $\grm = 1 - \grp$ so that $a = \frac{ \gs(\grm)}{\gs(\grp)} =1$ and the boundary conditions \eqref{eq: bords G2}, \eqref{eq: bords G1} simplify. 
In particular $G^2_t (x,y) =  G^1_t (x,y)$ and the Green's functions are associated to the Laplacian on the ring $[0,1]$
\begin{eqnarray*}
G^1_t (x,y) 
=  1  + 2 \sum_{k = 1}^\infty   \exp \big(  - (2 k \pi)^2 t \big) \cos ( 2 \pi k ( x- y) )  \, .
\end{eqnarray*}
As $G^1_t (y,z)  = G^1_t (z,y)$, one gets
\begin{eqnarray*}
&& \int_0^\infty ds  \int_0^1 dz \,   G^1_s (x,z)  G^1_s (z,y) -   G^1_s (x,0)  G^1_s (y,0) 
= \int_0^\infty ds  \,   G^1_{2 s} (x,y)  -   G^1_s (x,0)  G^1_s (y,0) \nonumber \\
&& = \sum_{k \geq 1} \frac{\cos \big(  2 \pi k (x-y) \big)}{(2 k \pi)^2}
- 2 \sum_{k = 1}^\infty  \frac{\cos( 2 \pi k x) + \cos( 2 \pi k y) }{(2 k \pi)^2}   
- 4 \sum_{k,n \geq 1}^\infty  \frac{\cos( 2 \pi n x) \cos( 2 \pi k y)}{(2 \pi)^2 (k^2 + n^2)}   
\, .
\end{eqnarray*}
From the identity 
\begin{eqnarray*}
\forall x\in [0,1], \qquad 
2 \sum_{k \geq 1} \frac{\cos \big(  2 \pi k x \big)}{(2 k \pi)^2}
= - \frac{1}{2 }  x (1-x)+  \frac{1}{12} \, .
\end{eqnarray*}
We deduce that for $x<y$
\begin{eqnarray*}
&& \int_0^\infty ds  \int_0^1 dz \,   G^1_s (x,z)  G^2_s (z,y) -   G^1_s (x,0)  G^1_s (y,0) 
= \int_0^\infty ds  \,   G^1_{2 s} (x,y)  -   G^1_s (x,0)  G^1_s (y,0) \nonumber \\
&& =  - \frac{1}{4}  (y-x) (1+x-y)+  \frac{1}{24}   
+ \frac{1}{2 }  x (1-x) + \frac{1}{2 }  y (1-y) - \frac{1}{6}
-  4 \sum_{k,n \geq 1}^\infty  \frac{\cos( 2 \pi n x) \cos( 2 \pi k y)}{(2 \pi)^2 (k^2 + n^2)}   
\, ,
\label{eq: cas 1/2}
\end{eqnarray*}
Finally when the mean density is equal to $1/2$ then  \eqref{eq: C(x,w)} can be rewritten for $x<y$
\begin{eqnarray}
\label{eq: theory}
C(x,y) = 
 - \frac{1}{2} \gs(\grp) + \frac{1}{12}  \cJ^2 -  \cJ^2 \left( \frac{1}{2} (x+y)(1 - (x+y)) + x  \right)
 + 8 \cJ^2 \sum_{k,n \geq 1}^\infty  \frac{\cos( 2 \pi n x) \cos( 2 \pi k y)}{(2 \pi)^2 (k^2 + n^2)}
 \, . \nonumber \\
\end{eqnarray}
We stress the fact that the correlation function diverges when both $x,y$ approach the battery.

 \begin{figure}[h]
 \begin{center}
 \leavevmode
 \epsfysize = 5.5 cm
 \epsfbox{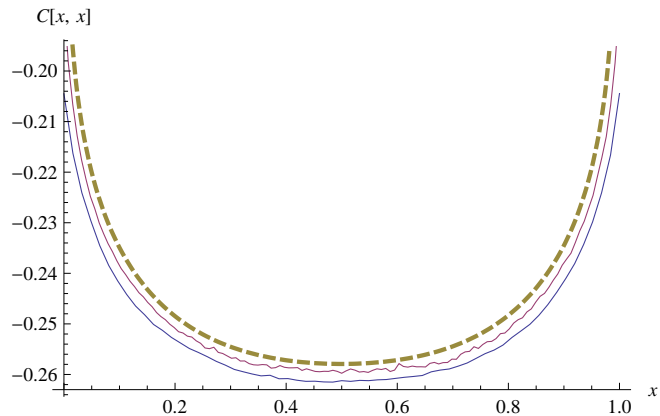}
 \end{center}
 \caption{At mean density $\bar \rho = 1/2$, the correlation $N \bra \eta_i ; \eta_{i+1} \ket$ between  the sites $i$ and $i+1$  versus $x=i/N$ 
 is represented for systems of sizes $64$ and $128$  with $25\times10^9$ updates per site. The dashed line represents
 the  theoretical prediction 
 \eqref{eq: theory}  of $C(x,x)$.}
 \label{fig: data6}
 \end{figure}

In figure \ref{fig: data6}, we compare the exact expression \eqref{eq: theory} to the results of a simulation, for $N=64$ and $N = 128$ for $25 \times 10^9$ updates per site. The agreement is very good. As our data for size $N=128$ lies between the data for $N=64$ and the theoretical prediction, it is reasonable to believe that the prediction does represent the $N \to \infty$ limit.

\section{Some exactly solvable models}
\label{sec: exactly solvable models}

\subsection{Zero range process}
\label{sec: Zero range process}

We study now the ZRP driven by a battery defined in section \ref{sec: moving battery}.
If $p =q$ then the dynamics is reversible wrt the product measure $\otimes_{i= 1}^N \nu_\phi$, where the measure at each site is given by
\begin{eqnarray*}
\nu_\phi (k) = \frac{1}{Z_\phi} \frac{\phi^k}{g(k) !}, \quad
{\rm with} \qquad g(k)! = \prod_{i=1}^k \; g(i)  \, ,
\end{eqnarray*}
where $Z_\phi$ is the normalization factor and we used the convention $g(0)! = 1$.

\medskip

We are going to show that for $p \not = q$, the invariant measure is a product measure with a canonical constraint. A similar property was already derived for zero range models in contact with reservoirs \cite{MF}.
For any $\gb$ such that $(q-p) \gb > 0$, we introduce 
\begin{eqnarray}
\label{eq: variation phi}
\bar \nu_\gb = \otimes_{i= 1}^N \nu_{\phi_i}, \qquad {\rm with} \qquad 
\phi_i = \gb \left( (i-1) +  {p (N-1) +1 \over q-p }  \right)  \, .
\end{eqnarray}
The parameter $\gb$ can be adjusted to fix the mean density.

\begin{pro}
\label{pro: ZR}
The invariant measure of the zero range process on a ring of length $N$ with $M$ particles
is the measure $\bar  \nu_\gb$ \eqref{eq: variation phi} 
conditioned to a total number of particles equal to $M$.
The conditioned measure is independent of $\gb$.
\end{pro}

\begin{proof}
The generator of the zero range process is given by
\begin{eqnarray*}
L f(\eta) & = \sum_{i \not = 1,N} g(\eta_i) \big[ (f( \eta^{(i,i+1)}) - f(\eta)) +  (f( \eta^{(i,i-1)}) - f(\eta))
\big] \\
& \qquad \qquad + g(\eta_1) \big[ q (f( \eta^{(1,N)}) - f(\eta)) +  (f( \eta^{(1,2)}) - f(\eta)) \big] \\
& \qquad \qquad + g(\eta_N) \big[ p (f( \eta^{(N,1)}) - f(\eta)) +  (f( \eta^{(N,N-1)}) - f(\eta)) \big] 
\end{eqnarray*}
where $\eta^{i,j}$ is the modified configuration after a jump from $i$ to $j$.

In order to check that $\bar \nu_\gb$ is an invariant measure, it is enough to prove that  for any function $f$, we have 
$\bar  \nu_\gb ( L f) =0$.
\begin{eqnarray*}
&& \bar \nu_\gb (L f) = \sum_\eta \bar \nu_\gb (\eta) \ \sum_{i \not=  N} g(\eta_i)  (f( \eta^{(i,i+1)}) - f(\eta)) + 
g(\eta_{i+1})  (f( \eta^{(i+1,i)}) - f(\eta))  \\
&& \qquad \qquad + q g(\eta_1) (f( \eta^{(1,N)}) - f(\eta)) + p g(\eta_N) (f( \eta^{(N,1)}) - f(\eta)) \\
&& = \sum_\eta \sum_{i \not=  N}  \bar  \nu_\gb (\eta) f(\eta) \Big[ 
g(\eta_{i+1}) \left( 1 + \frac{\phi_i - \phi_{i+1}}{\phi_{i+1}} \right) - g(\eta_{i+1}) 
+ g(\eta_i) \left( 1 + \frac{ \phi_{i+1} - \phi_i }{\phi_i} \right) - g(\eta_i)  \Big] \\
&& \qquad \qquad + 
\bar \nu_\gb (\eta) f(\eta) \Big[ q g(\eta_N)  \left( 1 + \frac{ \phi_1 - \phi_N }{\phi_N} \right)
- q g(\eta_1)  + p g(\eta_1) \left( 1 + \frac{ \phi_N - \phi_1}{\phi_1} \right)
- p g(\eta_N)  \Big] \\
&& = \sum_\eta \sum_{i \not=  N} \bar  \nu_\gb  (\eta) f(\eta) \;  (\phi_i - \phi_{i+1} ) \left[ 
 \frac{g(\eta_{i+1}) }{\phi_{i+1}}  -   \frac{ g(\eta_i)}{\phi_i}   \right] \\
&& \qquad \qquad + 
\bar \nu_\gb (\eta) f(\eta) \Big[  g(\eta_N)  \left( q-p + q\frac{ \phi_{1} - \phi_N}{\phi_N} \right)
 + g(\eta_1) \left( p - q + p \frac{ \phi_{N} - \phi_1}{\phi_1} \right)
 \Big] \\
\end{eqnarray*}
Since $\phi_i$ varies linearly according to \eqref{eq: variation phi}, we obtain 
\begin{eqnarray*}
\bar \nu_\gb (L f )  =  \bar  \nu_\gb \left(  f(\eta) \; \Big(  \gb \left[ 
 \frac{ g(\eta_1)}{\phi_1}   -  \frac{g(\eta_N) }{\phi_N}  \right] 
+  g(\eta_N)  \left( \frac{ q \phi_1 - p \phi_N }{\phi_N} \right)
 + g(\eta_1) \left(  \frac{ p \phi_N - q\phi_1}{\phi_1} \right)
 \Big) \right) = 0 \, .
\end{eqnarray*}
The last equality follows from \eqref{eq: variation phi} which implies that $q \phi_1 - p \phi_N = \gb$.

The number of particles being preserved by the dynamics, the invariant measure is obtained by conditioning $\bar  \nu_\gb$ to have $M$ particles.
\end{proof}

\medskip

Proposition \ref{pro: ZR} can be recovered from  the macroscopic approach. 
The diffusion and conductivity coefficients of the ZRP satisfy $D(\gr) = \gs' (\gr)/2$.
Therefore in the case of a slowly varying weak field, the non-equilibrium part of the correlations $C_\mmu$ in \eqref{C-sol} vanishes and only the canonical constraint remains.

\subsection{Independent variables and canonical constraints}
\label{sec: Independent variables}

We are going to compute the scaling limit of the truncated two-point correlations in product measures with canonical constraints.
By analogy with the measures found in the previous section, we consider a chain of $N$ independent variables $\{\eta_i\}_{0 \leq i \leq N}$ with densities $\gr_i = \bar \gr( i /N)$ which vary slowly from site to site. 

Let $\{ \nu_\gr \}_\gr$ be a family of measures with different chemical potentials which are indexed by their density $\gr$.
The variance of the measure $\nu_\gr$ will be denoted by $\chi(\gr)$. 
We consider the measure $\bar \mu$ which is the product measure $\otimes_{i=1}^N \nu_{\gr_i}$
 conditioned to have a total number of particles equal to the integer part of $\sum_{i = 1}^N \bar \gr( i /N)$. 
The density profile of the measure  $\bar \mu$ is given by $\bar \gr(x)$ and the two-point correlation function scales for large $N$ and $x \not = y$ in $[0,1]$ as
\begin{eqnarray}
\label{eq: correlation independant}
\bar \mu \big( \eta_{Nx} \, ; \,  \eta_{Ny} ) = - \frac{1}{N} \; {\chi \big( \bar \gr(x) \big) \chi \big( \bar \gr(y) \big)  \over \int_0^1 \chi(\bar \gr(r)) dr } \, .
\end{eqnarray}
These asymptotics are very different from the long range non-equilibrium correlations of the SSEP with a battery and the same density profile \eqref{eq: C(x,w)}.

\bigskip

To derive \eqref{eq: correlation independant}, one can approximate the original system in the large $N$ limit by 
a chain of Gaussian independent variables under the canonical constraint.
Define $\mu$ to be the Gaussian measure of $N$ independent Gaussian variables $\{X_i\}_{1 \leq i \leq N}$ with variance 
$\chi_i = \chi \big(  \bar \gr( i /N)  \big)$ at site $i$ under the canonical constraint $\sum_{i=1}^N X_i = 0$.

To compute the marginal of $\mu$ over the two variables $\{X_1,X_2\}$,
we integrate over the $N-2$ remaining Gaussian variables.
Given $\{X_1,X_2\}$ the probability density that $\sum_{i=3}^N X_i =  - X_1 -X_2$
is proportional to $\exp \big( - {1 \over 2\; \hat \chi} ( X_1 + X_2)^2 \big)$ with $\hat \chi = \sum_{i=3}^N  \chi_i$.
Thus  the marginal of $\mu$ over the two variables $\{X_1,X_2\}$ is Gaussian with correlation matrix
\begin{eqnarray*}
\left( 
\begin{array}{cc}
 {1 \over \chi_1}  + {1 \over \;  \hat \chi}   &  {1 \over \;  \hat \chi}\\
{1 \over \;  \hat \chi}    & {1 \over \chi_2}  + {1 \over \;  \hat \chi}  
\end{array}
\right)^{-1} 
= {1 \over \hat \chi + \chi_1 + \chi_2}
\left( 
\begin{array}{cc}
\chi_1 (\hat \chi+ \chi_2)    & -\chi_1 \chi_2 \\
-\chi_1 \chi_2  &   \chi_2 (\hat \chi+ \chi_1) 
\end{array}
\right)
\end{eqnarray*}
For large $N$, $\hat \chi \simeq N  \int_0^1 \chi(\bar \gr(r)) dr$, so that the two point correlation scales like  
\begin{eqnarray*}
\mu ( X_1 ; X_2) \simeq  - {1 \over N} \; {\chi_1 \chi_2 \over \int_0^1 \chi( \bar \gr(r)) dr } \, .
\end{eqnarray*}
Thus \eqref{eq: correlation independant} follows.

\section{Conclusion}

In this paper, we studied several NESS with a canonical constraint on the particle number and computed the two-point correlation functions by using a macroscopic approach. We considered two driving mechanisms: either a microscopic strong field localized on one bond (the battery model), or a weak field with an intensity which varies smoothly on the macroscopic scale. 
The battery model can be recovered as a singular limit of the smoothly varying field.
Different driving mechanisms lead to different stationary measures with different long range correlations. 
In fact, the long range correlations of the canonical models  differ also from the correlations in systems maintained out of equilibrium by reservoirs.

\bigskip

\noindent
{\bf Acknowledgments:}
We thank  P. K. Mohanty, S. Olla and E. R. Speer for helpful discussions.
JL acknowledges the support of  the NSF Grant DMR-044-2066 and AFOSR Grant AF-FA9550-04.
BD and TB acknowledge the support of the French Ministry of Education through the ANR BLAN07-2184264 grant.
The work of TB was partially supported by the NSF Grant DMR-044-2066 and AFOSR Grant AF-FA9550-04 during a stay at Rutgers University and by  the {\it Florence Gould Foundation Endowment} during a stay at the Institute for Advanced Study.

\section{Appendix I: Green's functions}

In this appendix, we derive the exact expressions (\ref{eq: G1 0}), (\ref{eq: G1,G2})
for the Green's functions $G^1$, $G^2$  which satisfy (\ref{eq: bords  G2}), (\ref{eq: bords  G1}).
We recall that $a = \frac{\gs(\grm)}{\gs(\grp)}$.

\medskip

We start with the computation of $G^1_t$ introduced in \eqref{eq: bords G1} which evolves according to 
\begin{eqnarray}
\partial_t G^1_t (x,y) = \partial_x^2 G^1_t (x,y) ,
\quad {\rm with} \qquad
G^1_t (1,y) = a  G^1_t (0,y), 
\quad   \partial_x G^1_t (1,y) = \partial_x G^1_t (0,y) \, ,
\label{eq: bords G1 0}
\end{eqnarray}
with $G^1_0 (x,y) = \delta_{x=y}$.
The Laplacian with boundary conditions \eqref{eq: bords G1 0} at the battery 
is  not self-adjoint, but it can be decomposed on the following basis.  
We define
\begin{eqnarray}
\label{eq: base 1}
f^1_k (x) = (a x + 1 -x) \cos(2 k \pi x),  \qquad 
f^2_k (x) =  \sin (2 k \pi x) \, .
\end{eqnarray}
These functions satisfy the boundary conditions \eqref{eq: bords G1 0}.
Furthermore
\begin{eqnarray*}
\partial_x^2 f^1_k (x) =  - (2 k \pi)^2 f^1_k (x) + 4 k \pi (1-a) f^2_k (x)       ,  \qquad 
\partial_x^2  f^2_k (x) = - (2 k \pi)^2 f^2_k (x) \, .
\end{eqnarray*}
Let $f^1_k (t,x), f^2_k (t,x)$ be the solution at time $t$ of the heat equation with boundary conditions \eqref{eq: bd conditions phi} and initial data $f^1_k (x), f^2_k (x)$, then 
\begin{eqnarray}
\label{eq: evolution Delta 1}
\begin{cases}
f^1_k (t,x) = \exp \big(  - (2 k \pi)^2 t \big) \Big( f^1_k (x) + 4 k \pi (1-a) t f^2_k (x) \Big)   \\
f^2_k (t,x) = \exp \big(  - (2 k \pi)^2 t \big)  f^2_k (x) 
\end{cases}
\end{eqnarray}

\bigskip

The Green's function  $G^2$ introduced in \eqref{eq: bords  G2} satisfies
\begin{eqnarray}
\partial_t G^2_t (x,y) = \partial_x^2 G^2_t (x,y) ,
\quad {\rm with} \qquad
G^2_t (1,y) =  G^2_t (0,y), \quad  a \partial_x G^2_t (1,y) = \partial_x G^2_t (0,y),
\label{eq: bords  G2 0}
\end{eqnarray}
with $G^2_0 (x,y) = \delta_{x=y}$.
For the Laplacian with boundary conditions \eqref{eq: bords  G2 0}, we can introduce the basis
\begin{eqnarray}
\label{eq: base 2}
g^1_k (x) =  \cos (2 k \pi x),  \qquad  g^2_k (x) = (a +  x  - a x) \sin(2 k \pi x) \, .
\end{eqnarray}
and one has
\begin{eqnarray*}
\partial_x^2 g^1_k (x) =  - (2 k \pi)^2 g^1_k (x)     ,  \qquad 
\partial_x^2  g^2_k (x) = - (2 k \pi)^2 g^2_k (x) + 4 k \pi (1-a) g^1_k (x)    \, .
\end{eqnarray*}
If $g^1_k (t,x), g^2_k (t,x)$ denote the solution at time $t$ of the heat equation with boundary conditions \eqref{eq: bords  G2 0} and initial data $g^1_k (x), g^2_k (x)$, then 
\begin{eqnarray}
\label{eq: evolution Delta 2}
\begin{cases}
g^1_k (t,x) = \exp \big(  - (2 k \pi)^2 t \big)  g^1_k (x)    \\
g^2_k (t,x) = \exp \big(  - (2 k \pi)^2 t \big) \Big( 4 k \pi (1-a) t g^1_k (x) + g^2_k (x) \Big) 
\end{cases}
\end{eqnarray}

Using the fact that for $x,y \in (0,1)$
\begin{eqnarray*}
\sum_{n=0}^\infty  \cos( 2 n \pi (x- y) ) = \frac{1}{2} + \frac{1}{2} \gd_{x=y}  \, .
\end{eqnarray*}
one can check that 
\begin{eqnarray*}
&& \frac{a y +1 - y}{2} + \frac{a+1}{4} \gd_{x=y} \\
&& \qquad = \sum_{k=0}^\infty (ay +1 -y) \cos( 2 k \pi y) \cos( 2 k \pi x)
+ (a +x - ax)  \sin( 2 k \pi y) \sin( 2 k \pi x) \, . \nonumber
 \\ && \qquad = \sum_{k=0}^\infty  \,  f_k^1(y) g_k^1 (x) + f_k^2(y) g_k^2(x) \, .
\end{eqnarray*}
This implies 
\begin{eqnarray}
\frac{a+1}{4} \gd_{x = y} =  \frac{1}{2}  f^1_0 (y)   g^1_0 (x) + \sum_{k = 1}^\infty  
f^1_k (y) g^1_k (x) + f^2_k (y) g^2_k (x)   \, .
\label{il fait froid}
\end{eqnarray}

\bigskip

Thus  the Green's function $G^1$ is  given by
\begin{eqnarray}
\label{eq: G1}
\frac{a+1}{4} G^1_t (x,y) &=& \frac{1}{2}   f^1_0 (x)   + \sum_{k = 1}^\infty    f^1_k (t, x) g^1_k (y) + f^2_k (t,x) g^2_k (y) \\
&=&   \frac{1}{2} f^1_0 (x)   + \sum_{k = 1}^\infty   \exp \big(  - (2 k \pi)^2 t \big)
\Big(    f^1_k (x) g^1_k (y) + f^2_k (x) g^2_k (y) + 4 k \pi (1-a) t f^2_k (x) g^1_k (y)   \Big)  \, . \nonumber
\end{eqnarray}

Similarly by exchanging $x$ and $y$ in the r.h.s. of (\ref{il fait froid}) one can see that
\begin{eqnarray}
\frac{a+1}{4} G^2_t (x,y) =  \frac{1}{2}  f^1_0 (y) + \sum_{k = 1}^\infty   \exp \big( - (2 k \pi)^2 t \big)
\Big(    f^1_k (y) g^1_k (x) + f^2_k (y) g^2_k (x) + 4 k \pi (1-a) t f^2_k (y) g^1_k (x)   \Big)  \, . \nonumber\\
\label{eq: G2}
\end{eqnarray}
This proves the expression \eqref{eq: G1 0}  and 
the identity $G^2_t (x,y) = G^1_t (y,x)$ claimed in \eqref{eq: G1,G2}.

\section{Appendix II}

In this appendix, we are going to derive the relation \eqref{eq: expansion}  from \eqref{eq: phi, h}.

\medskip

We first start by proving some identities.
The correlation between $x$ and $w$ is symmetric so that for any $x$ and $w$
\begin{eqnarray*}
 \int_0^1 dy \,   G^1_s (x,y) \partial_y \Big( \gs(\bar \gr (y) ) \partial_y  G^2_s (y,w)  \Big)
 =
 \int_0^1 dy \,   G^1_s (w,y) \partial_y \Big( \gs(\bar \gr (y) ) \partial_y  G^2_s (y,x)  \Big)
\, . \nonumber 
\end{eqnarray*}
Using the fact that $G^1_s (x,y) = G^2_s (y,x)$ \eqref{eq: G1,G2}, one gets
\begin{eqnarray}
\label{eq: sym}
\int_0^1 dy \,   G^1_s (x,y) \partial_y \Big( \gs(\bar \gr (y) ) \partial_y  G^2_s (y,w)  \Big)
= \int_0^1 dy \,   G^2_s (y,w) \partial_y \Big( \gs(\bar \gr (y) ) \partial_y  G^1_s (x,y)  \Big)
\, . 
\end{eqnarray}
Expanding each terms of \eqref{eq: sym} leads to the two identities, 
\begin{eqnarray}
\label{eq: 1st IPP sym}
&& \int_0^1 dy \,   G^1_s (x,y) \partial_y \Big( \gs(\bar \gr (y) ) \partial_y  G^2_s (y,w)  \Big)\\
&& \qquad = \int_0^1 dy \,    G^1_s (x,y) \partial_y  \gs(\bar \gr (y) ) \partial_y  G^2_s (y,w) 
+ \int_0^1 dy \,    G^1_s (x,y)  \gs(\bar \gr (y) ) \partial^2_y  G^2_s (y,w) \, , \nonumber \\
&& \qquad = \int_0^1 dy \,    \partial_y  G^1_s (x,y) \partial_y  \gs(\bar \gr (y) ) G^2_s (y,w) 
+ \int_0^1 dy \,    \partial^2_y  G^1_s (x,y)  \gs(\bar \gr (y) ) G^2_s (y,w) \, . \nonumber 
\end{eqnarray}

We introduce the notation $[F(y)]_{y = 0}^{y=1}=F(1) - F(0)$.
Integrating by parts the first term in the last equation of \eqref{eq: 1st IPP sym}, one gets
\begin{eqnarray}
&& 
\int_0^1 dy \,    \partial_y  G^1_s (x,y) \partial_y  \gs(\bar \gr (y) ) G^2_s (y,w)   \nonumber \\
&& \qquad =  
- \int_0^1 dy \,  G^1_s (x,y)   \partial_y \Big(  \partial_y \gs(\bar \gr (y) ) G^2_s (y,w) \Big) 
+ [ G^1_s (x,y)   \partial_y \gs(\bar \gr (y) ) G^2_s (y,w) ]_{y = 0}^{y=1} \nonumber  \\
&& \qquad = 
- \int_0^1 dy \, G^1_s (x,y)   \,  \partial_y \gs(\bar \gr (y) ) \, \partial_y G^2_s (y,w) 
- \int_0^1 dy \,  G^1_s (x,y)   \partial^2_y \gs(\bar \gr (y) ) G^2_s (y,w)  \nonumber \\
&& \qquad \qquad 
+ [ G^1_s (x,y)   \partial_y \gs(\bar \gr (y) ) G^2_s (y,w) ]_{y = 0}^{y=1} \, .
\label{eq: 2nd IIP sym}
\end{eqnarray}
By summing the two equalities in \eqref{eq: 1st IPP sym} and simplifying thanks to the identity 
\eqref{eq: 2nd IIP sym}, we finally obtain

\begin{eqnarray}
&& 2  \int_0^1 dy \,  G^1_s (x,y) \partial_y \Big( \gs(\bar \gr (y) ) \partial_y  G^2_s (y,w)  \Big) \nonumber \\
&&\qquad =  \int_0^1 dy \,   \partial^2_y  G^1_s (x,y)  \gs(\bar \gr (y) )  G^2_s (y,w)  
+ \int_0^1 dy \,   G^1_s (x,y)  \gs(\bar \gr (y) ) \partial^2_y  G^2_s (y,w) \nonumber \\
&& \qquad \qquad 
- \int_0^1 dy \,    G^1_s (x,y)   \partial^2_y \gs(\bar \gr (y) ) G^2_s (y,w) 
+
[ G^1_s (x,y)   \partial_y \gs(\bar \gr (y) ) G^2_s (y,w) ]_{y = 0}^{y=1} \, .
\label{eq: IPP res}
\end{eqnarray}
Note that $\partial^2_y \gs(\bar \gr (y) ) = -4  \cJ^2$, where $\cJ = \gr^+ - \gr^-$ is the macroscopic mean current
(for $N$ large $\cJ = N \hq$ with $\hq$ introduced in \eqref{eq: mean flux}). 
The discontinuity of the density at the battery implies
$ \partial_y \gs(\grm) - \partial_y \gs( \grp ) = - 4 (\gr^+ - \gr^-) \cJ = - 4 \cJ^2$.
Thus integrating \eqref{eq: IPP res} wrt time leads to
\begin{eqnarray}
\label{eq: appendice 2}
&& \int_0^T ds  \int_0^1 dy \,  G^1_s (x,y) \partial_y \Big( \gs(\bar \gr (y) ) \partial_y  G^2_s (y,w) \Big) \\
&& \qquad =    \int_0^1 dy \,  G^1_T (x,y) \frac{\gs(\bar \gr (y))}{2}  G^2_T (y,w)   -  \frac{1}{2} \gs(\bar \gr (x) )  \gd_{x = w} \nonumber \\
&& \qquad + 2 \cJ^2 \, \int_0^T ds  \int_0^1 dy \,   G^1_s (x,y)  G^2_s (y,w) 
- 2  \cJ^2 \int_0^T ds \,  G^1_s (x,0)  G^2_s (0,w) \, . \nonumber
\end{eqnarray}
Letting $T$ go to infinity, one gets from the expression of $G^1$ \eqref{eq: G1 0} and $G^2$ \eqref{eq: G1,G2}
\begin{eqnarray*}
\lim_{T \to \infty} \int_0^1 dy \,  G^1_T (x,y) \frac{\gs(\bar \gr (y))}{2}  G^2_T (y,w)   
= \frac{2}{(a+1)^2} \ga (x) \ga (w) \int_0^1 dy \,   \gs(\bar \gr (y))    \, ,
\end{eqnarray*}
with the notation $a =  \frac{\gs( \grm)}{\gs( \grp )}$ and $\ga (x) = (ax + 1-x)$.

As the steady state is linear, 
$\gs(\bar \gr (y)) = \gs(\grp) +y (  \gs(\grm) - \gs(\grp) +2 \cJ^2) -2  \cJ^2 y^2$, and we get
\begin{eqnarray*}
\lim_{T \to \infty} \int_0^1 dy \,  G^1_T (x,y) \frac{\gs(\bar \gr (y))}{2}  G^2_T (y,w)   
= \frac{2}{(a+1)^2} \ga (x) \ga (w)  \left(  \frac{\gs(\grp ) + \gs(\grm)}{2} +  \frac{\cJ^2}{3} \right) \, .
\end{eqnarray*}
Using \eqref{eq: appendice 2} and the symmetry $G^1_s (w,z) = G^2_s (z,w)$ \eqref{eq: G1,G2}, we have finally shown that
\begin{eqnarray*}
\lim_{T \to \infty}  \int_0^T ds  \int_0^1 dy \,  G^1_s (x,y) \partial_y \Big( \gs(\bar \gr (y) ) \partial_y  G^2_s (y,w) \Big) 
 =      -  \frac{1}{2} \gs(\bar \gr (x) )  \gd_{x = w} - C(x,w) \, ,
\end{eqnarray*}
with
\begin{eqnarray}
 C(x,w) &=&  
 - \frac{2}{(a+1)^2}  \left( \frac{\gs(\grp ) + \gs(\grm)}{2} +  \frac{\cJ^2}{3} \right)  (ax + 1-x) (a w + 1-w) \\
&& \qquad -  2 \cJ^2 \, \int_0^\infty ds  \int_0^1 dz \,   G^1_s (x,z)  G^1_s (w,z) -   G^1_s (x,0)  G^1_s (w,0) \, . \nonumber
\end{eqnarray}

\medskip

The previous expression combined to \eqref{eq: phi, h} leads to \eqref{eq: contrainte }.

\section{Appendix III}

\subsection{Uniqueness of the steady state}
\label{sec: Uniqueness of the steady state}

We will show the uniqueness of the smooth solutions of \eqref{steady-state}
\begin{equation}
\label{steady-state 2}
\partial_x \Big( D \big( \bar \gr(x) \big) \bar \gr' (x) - \mmu(x) \gs \big( \bar \gr(x) \big) \Big) =  0 \, ,
\end{equation}
for regular coefficients and $D$ positive.

Suppose that $\gr_1(x),\gr_2(x)$ are two solutions of \eqref{steady-state 2} with the same mean density.
Then there is $\cJ_1 > \cJ_2$ such that 
\begin{eqnarray*}
-D \big( \bar \gr_1(x) \big) \bar \gr_1 ' (x) + \mmu(x) \gs \big( \bar \gr_1 (x) \big) &=&  \cJ_1 \\
-D \big( \bar \gr_2 (x) \big) \bar \gr_2 ' (x) + \mmu(x) \gs \big( \bar \gr_2 (x) \big) &=&  \cJ_2 
\end{eqnarray*}
If the profiles coincide at $x_0$ then substracting both equations, one has
\begin{eqnarray*}
D \big( \bar \gr_1(x_0) \big) [ \bar \gr_1 ' (x_0) - \bar \gr_2 ' (x_0) ] = \cJ_2 -   \cJ_1 < 0 \, . 
\end{eqnarray*}
This implies that when the two solutions cross then $\bar \gr_1 > \bar \gr_2$ before the crossing and $\bar \gr_1 < \bar \gr_2$ after the crossing, so that the solutions cannot cross more than once. 
But the profiles are smooth, periodic  and have the same mean density, thus they have to cross an even number of times.
This is a contradiction and therefore both profiles have to be equal ($\cJ_1 = \cJ_2$).

\subsection{Convergence of the linearized evolution}
\label{sec: Convergence of the linearized evolution}

We turn now to the convergence of the linearized evolution \eqref{eq: nouveau G1}
\begin{eqnarray}
\partial_t f (t,x) 
=
\partial_x^2 
 \Big( D(\bar \gr(x)) \; f (t,x) \Big) -  \partial_x \Big( \mmu(x)  \sigma' (\bar \gr(x))  f (t,x) \Big) \  ,
\label{eq: nouveau eq}
\end{eqnarray}
with smooth initial data $f(0,x)$ and mean $\int_0^1 dx \, f(0,x) = 1$. 
It is equivalent to consider the evolution
\begin{eqnarray}
\partial_t f (t,x) = \partial_x \Big( D(\bar \gr(x)) \; \partial_x f (t,x) \Big) -  \partial_x \Big( \ga(x)   f (t,x) \Big) \  ,
\label{eq: nouveau eq 2}
\end{eqnarray}
with $\ga(x) = - \partial_x  D(\bar \gr(x)) +  \mmu(x)  \sigma' (\bar \gr(x))$. 
At any time $t>0$, $f(t,x)$ can then be interpreted as the probability density of a particle evolving on the ring $[0,1]$ with  a non homogeneous diffusion coefficient $D(\bar \gr(x))$ and  a drift $\ga(x)$.
As the particle evolves on a compact set, it will reach a stationary state when $t$ goes to infinity.
The limiting density for the particle position will be denoted by $\bar f(x)$  and it is the unique solution with mean density $\int_0^1 dx \, \bar f (x) = 1$ of 
\begin{eqnarray}
 \partial_x \Big( D(\bar \gr(x)) \; \partial_x \bar f (x) \Big) -  \partial_x \Big( \ga(x)   \bar f (x) \Big) = 0 \, .
\label{eq: nouveau steady}
\end{eqnarray}

A way to understand the relaxation to the stationary state is to check that  the relative entropy
\begin{eqnarray}
S(t) = - \int_0^1 dx \; f(t,x) \log \left( {f(t,x) \over \bar f(x)} \right) \, ,
\label{eq: Lyapunov}
\end{eqnarray}
is a Lyapunov function for the evolution \eqref{eq: nouveau eq 2}. 
This is a general fact for Markov processes. 
Taking the time derivative one has
\begin{eqnarray}
\partial_t S(t) =  \int_0^1 dx \; \Big( D(\bar \gr(x)) \; \partial_x f (t,x) - \ga(x)   f (t,x) \Big)\;  \partial_x \log \left( {f(t,x) \over \bar f(x)} \right) \, .
\label{eq: Lyapunov 2}
\end{eqnarray}
We now note that 
\begin{eqnarray*}
\int_0^1 dx \;  \ga(x)   f (t,x)  \partial_x \log \left( {f(t,x) \over \bar f(x)} \right) 
&=& 
\int_0^1 dx \;  \ga(x)   \bar f (x)  \partial_x  \left( {f(t,x) \over \bar f(x)} \right) \nonumber \\ 
&=& 
\int_0^1 dx \; D(\bar \gr(x)) \; \partial_x \bar f (x)   \partial_x  \left( {f(t,x) \over \bar f(x)} \right) \, ,
\end{eqnarray*}
where we used \eqref{eq: nouveau steady} in the last equality.
Combined with \eqref{eq: Lyapunov 2}, this leads to
\begin{eqnarray}
\partial_t S(t) &=&  \int_0^1 dx \; D(\bar \gr(x)) \; \Big(  {\bar f(x) \over f(t,x)} \partial_x f (t,x) - \partial_x \bar f (x)\Big) 
  \partial_x  \left( {f(t,x) \over \bar f(x)} \right)  \nonumber \\
&=& 
4 \int_0^1 dx \; D(\bar \gr(x)) \bar f(x) \; \left(  \partial_x  \sqrt{ {f(t,x) \over \bar f(x)} } \right)^2 \, .
\label{eq: Lyapunov 3}
\end{eqnarray}
Thus $S(t)$ is a Lyapunov function.
A quantitative estimate of the approach to equilibrium could then be obtained by using  a log-sobolev inequality.

\end{document}